\documentclass{article}
\usepackage{lmodern}

\usepackage[T1]{fontenc}
\usepackage[latin9]{inputenc}
\usepackage{geometry}
\usepackage{color}
\usepackage[american,english]{babel}
\usepackage{array}
\usepackage{verbatim}
\usepackage{float}
\usepackage{url}
\usepackage{amsmath}
\usepackage{amsthm}
\usepackage{amssymb}
\usepackage{graphicx}
\usepackage[round]{natbib}
\usepackage[pdfusetitle,
 bookmarks=true,bookmarksnumbered=false,bookmarksopen=false,
 breaklinks=false,pdfborder={0 0 0},pdfborderstyle={},backref=false,colorlinks=false]
 {hyperref}

\makeatletter

\providecommand{\tabularnewline}{\\}
\floatstyle{ruled} 
\newfloat{algorithm}{tbp}{loa}
\providecommand{\algorithmname}{Algorithm}
\floatname{algorithm}{\protect\algorithmname}

\theoremstyle{plain}
\newtheorem{thm}{\protect\theoremname}[section]
\theoremstyle{definition}
\newtheorem{defn}[thm]{\protect\definitionname}
\theoremstyle{remark}
\newtheorem{rem}[thm]{\protect\remarkname}
\theoremstyle{plain}
\newtheorem{cor}[thm]{\protect\corollaryname}
\theoremstyle{plain}
\newtheorem{prop}[thm]{\protect\propositionname}
\theoremstyle{definition}
\newtheorem{example}[thm]{\protect\examplename}
\theoremstyle{plain}
\newtheorem{lem}[thm]{\protect\lemmaname}
\newenvironment{lyxlist}[1]
	{\begin{list}{}
		{\settowidth{\labelwidth}{#1}
		 \setlength{\leftmargin}{\labelwidth}
		 \addtolength{\leftmargin}{\labelsep}
		 }}
	{\end{list}}

\numberwithin{equation}{section}
\newtheorem{assumption}{Assumption}

\makeatother

\addto\captionsamerican{\renewcommand{\algorithmname}{Algorithm}}
\addto\captionsamerican{\renewcommand{\corollaryname}{Corollary}}
\addto\captionsamerican{\renewcommand{\definitionname}{Definition}}
\addto\captionsamerican{\renewcommand{\examplename}{Example}}
\addto\captionsamerican{\renewcommand{\lemmaname}{Lemma}}
\addto\captionsamerican{\renewcommand{\propositionname}{Proposition}}
\addto\captionsamerican{\renewcommand{\remarkname}{Remark}}
\addto\captionsamerican{\renewcommand{\theoremname}{Theorem}}
\addto\captionsenglish{\renewcommand{\corollaryname}{Corollary}}
\addto\captionsenglish{\renewcommand{\definitionname}{Definition}}
\addto\captionsenglish{\renewcommand{\examplename}{Example}}
\addto\captionsenglish{\renewcommand{\lemmaname}{Lemma}}
\addto\captionsenglish{\renewcommand{\propositionname}{Proposition}}
\addto\captionsenglish{\renewcommand{\remarkname}{Remark}}
\addto\captionsenglish{\renewcommand{\theoremname}{Theorem}}
\providecommand{\corollaryname}{Corollary}
\providecommand{\definitionname}{Definition}
\providecommand{\examplename}{Example}
\providecommand{\lemmaname}{Lemma}
\providecommand{\propositionname}{Proposition}
\providecommand{\remarkname}{Remark}
\providecommand{\theoremname}{Theorem}

\begin{document}
\title{From characteristic functions to multivariate distribution functions
and European option prices by the (damped) COS method}
\author{Gero Junike\thanks{\textcolor{red}{This is a post-peer-review, pre-copyedit version of an article published in SIAM Journal on Numerical Analysis. The final authenticated version is available online at:
Junike, G. and Stier, H. (2025). From Characteristic Functions to Multivariate Distribution Functions and European Option Prices by the (Damped) COS Method. SIAM Journal on Numerical Analysis, 63(6), 2421-2453, https://doi.org/10.1137/24M1666240}\\ Corresponding author. Department of Mathematics, Ludwig-Maximilians
Universität, Theresienstr. 39, 80333 München, Germany. ORCID: 0000-0001-8686-2661,
E-mail: gero.junike@math.lmu.de}, Hauke Stier\thanks{Department of Mathematics, Carl von Ossietzky Universität, 26129 Oldenburg,
Germany.}}
\maketitle
\begin{abstract}
We provide a unified framework to obtain numerically certain quantities,
such as the distribution function, absolute moments and prices of
financial options, from the characteristic function of some (unknown)
probability density function using the Fourier-cosine series (COS)
method. The classical COS method is numerically very efficient in
one dimension, but it cannot deal very well with certain integrands
in general dimensions. Therefore, we introduce the damped COS method,
which can handle a large class of integrands very efficiently. We
prove the convergence of the (damped) COS method and study its order
of convergence. The method converges exponentially if the characteristic
function decays exponentially. To apply the (damped) COS method, one
has to specify two parameters: a truncation range for the multivariate
density and the number of terms to approximate the truncated density
by a Fourier-cosine series. We provide an explicit formula for the
truncation range and an implicit formula for the number of terms.
Numerical experiments up to five dimensions confirm the theoretical
results.

\textbf{Keywords: }Fourier transform, numerical integration, inversion
theorem, COS method, CDF, option pricing\\
 \textbf{Mathematics Subject Classification} 65D30 · 65T40 · 65Z05
· 60E10
\end{abstract}

\section{Introduction}
\author{We aim to solve the following integral numerically:
\begin{equation}
\int_{\mathbb{R}^{d}}w(\boldsymbol{x})g(\boldsymbol{x})d\boldsymbol{x}.\label{eq:int}
\end{equation}
The function $g$ is usually a density and the function $w$ is called
\emph{function of interest.} Integrals as in (\ref{eq:int}) appear
in a wide range of applications: The integral is equal to the cumulative
distribution function (CDF) of the density $g$ if $w$ is an indicator
function. CDFs appear in many scientific disciplines. }

If $w$ is the absolute value, the integral describes the absolute
moment of the density $g$, which plays an important role in various
disciplines but is not easy to obtain, see \cite{von1965convergence, brown1972formulae, barndorff2005absolute}
and references therein. 

In a financial context, the function $w$ might also describe some
financial contract, which depends on several assets. The function
$g$ is then the density of the logarithmic returns of the assets
and the integral describes the price of the contract. 

In many cases, the precise structure of $g$ is unknown, but the Fourier
transform $\widehat{g}$ is often given in closed form. For example:
while the joint density of the sum of two independent random variables
can only be expressed as a convolution and is usually not given explicitly,
the joint characteristic function is much simpler to obtain (it is
just the product of the marginal characteristic functions). Moreover,
the characteristic function of a Lévy process at a particular time-point
is usually given explicitly thanks to the Lévy-Khintchine formula.

The integral in (\ref{eq:int}) can be solved numerically using various
techniques, including (quasi) Monte Carlo simulation, numerical quadrature
and Fourier inversion. Special Fourier-inversion methods exist in
the case of a CDF, e.g., the Gil-Pelaez formula, see \cite{gil1951note}
and extensions, e.g., \cite{schorr1975numerical,waller1995obtaining,hughett1998error}
in one dimension and \cite{shephard1991characteristic,shephard1991numerical}
in $d$ dimensions.

The COS method, see \citet{fang2009novel} for $d=1$ and \cite{ruijter2012two}
for $d=2$, is a Fourier inversion technique. The COS method has been
applied extensively in computational finance, see \citet{fang2009pricing,fang2011fourier,grzelak2011heston,zhang2013efficient,leitao2018data,liu2019neural,liu2019pricing,oosterlee2019mathematical}
and \citet{bardgett2019inferring}.

The COS method has also been applied to solve backward stochastic
differential equations, see \cite{ruijter2015fourier} and \cite{andersson2023convergence}.

The main ideas of the COS method are to truncate the integration range
in (\ref{eq:int}) to some finite hypercube and to approximate the
density $g$ on the finite truncation range by a classical Fourier-cosine
series. There is a clever trick to approximate the Fourier-cosine
coefficients for $g$ in a very fast and robust way using $\widehat{g}$.
The COS method is particularly fast when the Fourier-cosine coefficients
of the function of interest $w$ are given analytically. However,
in multivariate dimensions this is rarely the case an exception being
a CDF. For example, the Fourier-cosine coefficients of arithmetic
basket options, call-on-minimum and spread options are not given explicitly.
\cite{ruijter2012two} propose in these cases to obtain the Fourier-cosine
coefficients of the function of interest numerically by a discrete
cosine transform, but this approach slows the COS method significantly.
In this article, we introduce the \emph{damped COS method}, which
is able to avoid the expensive application of the discrete cosine
transform if the Fourier transform of the function of interest is
given in closed form. The main idea is to damp $w$ by multiplying
it by an exponential function in order to make the damped function
of interest integrable. The idea of introducing a damping factor dates
back at least to \cite{carr1999option} and \citet{lewis2001simple}.
In multivariate dimensions, damping is also applied by \cite{eberlein2010analysis},
where the optimal damping factor is determined by \cite{bayer2022optimal}.

Other Fourier techniques to solve integrals as in (\ref{eq:int})
can be found in \citet{carr1999option,lewis2001simple,lord2008fast,eberlein2010analysis,ortiz2013robust,kirkby2015efficient,ortiz2016highly,baschetti2022sinc}
and \citet{bayer2024quasi}.

This article makes the following main contributions: We prove the
convergence of the multidimensional classical and damped COS methods,
we analyze the order of convergence of the classical and the damped
COS methods, and we provide explicit and implicit formulas for the
truncation range and the number of terms, respectively. The new approach
to finding the number of terms in $d$ dimensions is completely distinct
from the one-dimensional case discussed in \citet{junike2023handle}.
Unlike \cite{ruijter2012two}, who analyze the classical COS method,
we include in our analysis numerical uncertainty on the characteristic
function $\widehat{g}$. This makes it possible to understand how
approximations on $\widehat{g}$ affect the total error of the COS
method. 

In one dimension, it is well known that the COS method converges exponentially
if the density $g$ is smooth and has semi-heavy tails and performs
favorably to the Carr-Madan formula, see \citet{fang2009novel} and
\citet{junike2023handle}. However, the order of convergence of the
COS method is only linear if $g$ has heavy tails, which holds for
instance for the stable law. In such cases the Carr-Madan formula
outperforms the COS method, see \citet{junike2023handle}. There are
also other methods, such as the SINC method or the PROJ method, which
have demonstrated superior performance compared to the COS method
in certain scenarios, see \citet{baschetti2022sinc} and \citet{kirkby2015efficient}. 

In multivariate dimensions, the classical COS method becomes slow
for many function of interest $w$, e.g., arithmetic basket options,
since the Fourier-cosine coefficients of $w$ need to be approximated
numerically. In two dimensions, the classical COS method is significantly
slower than the Lewis formula, see \citet{bayer2022optimal}. This
disadvantage can be circumvented by the damped COS method introduced
in Section \ref{subsec:Comparison-with-Monte}.

The classical and damped COS methods suffer from the curse of dimensionality
and are efficient only for moderate dimensions: We show empirically
that the COS method outperforms a crude Monte Carlo simulation in
$d\in\{2,3\}$. A Monte Carlo simulation is recommended for $d\geq5$.
The choice between the COS method and a Monte Carlo simulation in
$d=4$ dimensions depends on the required error tolerance and the
function of interest. Note that some Fourier inversion approaches
are faster than Monte Carlo simulations up to 15 dimensions, see \citet{bayer2024quasi}.

This article is structured as follows: In Section \ref{sec:Notation}
we fix some notation. In Section \ref{sec:Damped-COS-method} we introduce
the multidimensional classical and damped COS methods, prove its convergence,
analyze the order of convergence and provide explicit and implicit
formulas for the truncation range and the number of terms. In Section
\ref{sec:Characteristic-functions} we discuss some examples for $g$
and $\widehat{g}$. In Section \ref{sec:Functions-of-interest} we
discuss some functions of interest, i.e., examples for $w$. In Section
\ref{sec:Numerical} we provide numerical experiments. Section \ref{sec:Conclusions}
concludes.

\section{\protect\label{sec:Notation}Notation}

Let $d\in\mathbb{N}$. Let $\mathcal{L}^{1}$ and $\mathcal{L}^{2}$
denote the sets of integrable and square-integrable functions from
$\mathbb{R}^{d}$ to $\mathbb{R}$ and by $\left\langle .,.\right\rangle $
and $\left\Vert .\right\Vert _{2}$ we denote the scalar product and
the (semi)norm on $\mathcal{L}^{2}$, respectively. The supremum norm
of a function $g:\mathbb{R}^{d}\to\mathbb{C}$ is defined by $\left\Vert g\right\Vert _{\infty}:=\sup_{\boldsymbol{x}\in\mathbb{R}^{d}}|g(\boldsymbol{x})|$.
By $\mathcal{L}^{\infty}$ we denote the set of bounded functions
from $\mathbb{R}^{d}$ to $\mathbb{R}$. By $\Re\{z\}$ and $\Im\{z\}$
we denote the real and imaginary parts of a complex number $z\in\mathbb{C}$.
The complex unit is denoted by $i$. By $\Gamma$, we denote the Gamma
function. The Euclidean norm and the maximum norm on $\mathbb{R}^{d}$
are denoted by $|.|$ and by $|.|_{\infty}$, respectively. For $\boldsymbol{x},\boldsymbol{y}\in\mathbb{R}^{d}$
we define
\[
\boldsymbol{x}\geq\boldsymbol{y}:\Leftrightarrow x_{1}\geq y_{1},...,x_{d}\geq y_{d}
\]
and treat ``$\leq$'', ``$<$'', ``$>$'', ``$=$'' and ``$\neq$''
similarly. We set $\mathbb{R}_{+}^{d}:=\{\boldsymbol{x}\in\mathbb{R}^{d},\boldsymbol{x}>\boldsymbol{0}\}$
and $\mathbb{R}_{-}^{d}:=\{\boldsymbol{x}\in\mathbb{R}^{d},\boldsymbol{x}<\boldsymbol{0}\}$.
For $\boldsymbol{a},\boldsymbol{b}\in\mathbb{R}^{d}$ with $\boldsymbol{a}\leq\boldsymbol{b}$,
two complex vectors $\boldsymbol{z},\boldsymbol{y}\in\mathbb{C}^{d}$
and $\lambda\in\mathbb{C}$ we define $\boldsymbol{z}+\boldsymbol{y}:=(z_{1}+y_{1},...,z_{d}+y_{d})\in\mathbb{C}^{d}$
and treat $\boldsymbol{zy}$ and $\frac{\boldsymbol{z}}{\boldsymbol{y}}$
similarly. We further define
\begin{align*}
\boldsymbol{z}\cdot\boldsymbol{y}:= & z_{1}y_{1}+...+z_{d}y_{d}\in\mathbb{C}\\
\lambda\boldsymbol{z}:= & (\lambda z_{1},...,\lambda z_{d})\in\mathbb{C}^{d}\\{}
[\boldsymbol{a},\boldsymbol{b}]:= & [a_{1},b_{1}]\times...\times[a_{d},b_{d}]\subset\mathbb{R}^{d}\\
(-\boldsymbol{\infty},\boldsymbol{b}]:= & (-\infty,b_{1}]\times...\times(-\infty,b_{d}]\subset\mathbb{R}^{d}\\
\exp(\boldsymbol{x}):= & \left(\exp(x_{1}),....,\exp(x_{d})\right)\in\mathbb{R}^{d},\quad\boldsymbol{x}\in\mathbb{R}^{d}\\
\log(\boldsymbol{x}):= & \left(\log(x_{1}),....,\log(x_{d})\right)\in\mathbb{R}^{d},\quad\boldsymbol{x}\in\mathbb{R}_{+}^{d}.
\end{align*}
For a subset $A\subset\mathbb{R}^{d},$ we define the indicator function
$1_{A}(\boldsymbol{x})$ by one if $\boldsymbol{x}\in A$ and by zero
otherwise. Let $\mathbb{N}_{0}=\mathbb{N}\cup\{0\}$. For $\boldsymbol{N}=(N_{1},...,N_{d})\in\mathbb{N}_{0}^{d}$
and a sequence $\left(a_{\boldsymbol{k}}\right)_{\boldsymbol{k}\in\mathbb{N}_{0}^{d}}\subset\mathbb{C}$,
we define
\begin{align*}
\sideset{}{'}\sum_{\boldsymbol{\boldsymbol{0}}\leq\boldsymbol{k}\leq\boldsymbol{N}}a_{\boldsymbol{k}} & :=\sum_{\boldsymbol{\boldsymbol{0}}\leq\boldsymbol{k}\leq\boldsymbol{N}}\frac{1}{2^{\Lambda(\boldsymbol{k})}}a_{\boldsymbol{k}},
\end{align*}
where $\Lambda(\boldsymbol{k})$ is the number of components of the
vector $\boldsymbol{k}$ that are equal to zero, i.e., $\Lambda(\boldsymbol{k}):=\sum_{h=1}^{d}1_{\{0\}}(k_{h})$.
For an integrable function $\psi:\mathbb{R}^{d}\to\mathbb{C}$ we
define its \emph{Fourier transform} by
\begin{equation}
\widehat{\psi}(\boldsymbol{u}):=\int_{\mathbb{R}^{d}}\psi(\boldsymbol{x})e^{i\boldsymbol{u}\cdot\boldsymbol{x}}d\boldsymbol{x},\quad\boldsymbol{u}\in\mathbb{R}^{d}.\label{eq:FourierTransform}
\end{equation}
This definition of the Fourier transform also appears in \citet[Def. 22.6]{bauer1996probability}
and in \cite{eberlein2010analysis}. Provided the integral in (\ref{eq:FourierTransform})
exists, the domain of $\widehat{\psi}$ may also be extended to parts
of the complex plane. If $\psi\geq0$ and $\int\psi(\boldsymbol{x})d\boldsymbol{x}=1$,
then $\psi$ is called \emph{density}, $\widehat{\psi}$ is called
the \emph{characteristic function} and the map $\boldsymbol{y}\mapsto\int_{(-\boldsymbol{\infty},\boldsymbol{y}]}\psi(\boldsymbol{x})d\boldsymbol{x}$
is called the \emph{cumulative distribution function }(CDF)\emph{.} 

The next definition introduces the concept of \emph{COS-admissiblility}
and extends \citet[Def. 1]{junike2022precise} to the multidimensional
setting. It is necessary to prove the convergence of the COS method. 
\begin{defn}
\label{def:COS_admissible}Let $\boldsymbol{L}=(L_{1},...,L_{d})\in\mathbb{R}_{+}^{d}$.
A function $\psi\in\mathcal{L}^{1}$ is called \emph{COS-admissible}
if
\[
B_{\psi}(\boldsymbol{L}):=\sideset{}{'}\sum_{\boldsymbol{k}\in\mathbb{N}_{0}^{d}}\frac{1}{\prod_{h=1}^{d}L_{h}}\left|\int_{\mathbb{R}^{d}\setminus[-\boldsymbol{L},\boldsymbol{L}]}\psi(\boldsymbol{x})e_{\boldsymbol{k}}(\boldsymbol{x})d\boldsymbol{x}\right|^{2}\to0,\quad\min_{h=1,...,d}L_{h}\to\infty.
\]
\end{defn}

The functions $e_{\boldsymbol{k}}$ are defined in Equation (\ref{eq:ek}).
If a function $\psi\in\mathcal{L}^{1}\cap\mathcal{L}^{2}$ satisfies
\begin{equation}
\int_{\mathbb{R}^{d}}\left|\boldsymbol{x}\right|^{2d}|\psi(\boldsymbol{x})|^{2}d\boldsymbol{x}<\infty,\label{eq:propBed}
\end{equation}
then $\psi$ is COS-admissible, see Proposition \ref{prop:COSadmissible}
in the appendix. In particular, it follows that bounded densities
with existing moments satisfy Inequality (\ref{eq:propBed}). 
\begin{rem}
In one dimension, there are other sufficient conditions than Inequality
(\ref{eq:propBed}) to ensure COS-admissiblility, see \citet[Corollary 4]{junike2022precise}
and \citet[Lemma 3.7]{junike2023handle}.
\end{rem}

\cite{junike2022precise} and \cite{junike2023handle} assume that
$g$ has semi-heavy tails, i.e., $g$ decays exponentially or faster.
Here, we make the same assumption in multivariate dimensions in order
to be able to estimate the truncation range and the number of terms
of the COS method.
\begin{defn}
We say a function $\psi:\mathbb{R}^{d}\to\mathbb{R}$ \emph{decays
exponentially} if there are $C_{1},C_{2},m>0$ such that for $|\boldsymbol{x}|>m$
it holds that $|\psi(\boldsymbol{x})|\leq C_{1}e^{-C_{2}|\boldsymbol{x}|}$.
\end{defn}

We denote the set of functions from $\mathbb{R}^{d}$ to $\mathbb{R}$
which decay exponentially and satisfy Inequality (\ref{eq:propBed})
by $\mathcal{L}^{e}$. 

\section{\protect\label{sec:Damped-COS-method}Classical and damped COS methods}

Problem statement: Throughout the paper, we fix a density $g:\mathbb{R}^{d}\to\mathbb{R}$
and a function of interest $w:\mathbb{R}^{d}\to\mathbb{R}$. We aim
to solve the integral in (\ref{eq:int}) numerically. 

Typically, for the function of interest, it holds that $w\notin\mathcal{L}^{1}$.
For example, the integral in (\ref{eq:int}) is equal to the CDF of
$g$ evaluated at $\boldsymbol{y}\in\mathbb{R}^{d}$ if $w(\boldsymbol{x})=1_{(\boldsymbol{-\infty},\boldsymbol{y}]}(\boldsymbol{x})$
for $\boldsymbol{x}\in\mathbb{R}^{d}$. Or, in a financial context,
an unweighted arithmetic basket put option is defined by $w(\boldsymbol{x})=\max(K-\sum_{h=1}^{d}e^{x_{h}},0)$,
$\boldsymbol{x}\in\mathbb{R}^{d}$, where $K>0$, see e.g., \citet{bayer2022optimal}. 

Note that for many densities and functions of interest, both $\widehat{g}$
and $\widehat{w}$ are given in closed form, see e.g., \cite{ruijter2012two,eberlein2010analysis}
and references therein. The density $g$ does not need to be known
precisely to apply the classical or the damped COS method. Let $\boldsymbol{\alpha}\in\mathbb{R}^{d}$
be a damping factor. Let
\begin{align}
\delta_{g}:= & \left\{ \boldsymbol{\alpha}\in\mathbb{R}^{d}:\,\boldsymbol{x}\mapsto g(\boldsymbol{x})e^{\boldsymbol{\alpha}\cdot\boldsymbol{x}}\in\mathcal{L}^{1}\cap\mathcal{L}^{2}\cap\mathcal{L}^{e}\right\} \label{eq:delta_g}\\
\delta_{w}^{\infty}:= & \left\{ \boldsymbol{\alpha}\in\mathbb{R}^{d}:\,\boldsymbol{x}\mapsto w(\boldsymbol{x})e^{-\boldsymbol{\alpha}\cdot\boldsymbol{x}}\in\mathcal{L}^{\infty}\right\} \label{eq:delta_w_inf}\\
\delta_{w}:= & \left\{ \boldsymbol{\alpha}\in\mathbb{R}^{d}:\,\boldsymbol{x}\mapsto w(\boldsymbol{x})e^{-\boldsymbol{\alpha}\cdot\boldsymbol{x}}\in\mathcal{L}^{1}\cap\mathcal{L}^{2}\cap\mathcal{L}^{e}\cap\mathcal{L}^{\infty}\right\} .\label{eq:delta_w}
\end{align}

The following Assumption \ref{A1} is a sufficient condition to prove
the convergence of the classical COS method and Assumption \ref{A2}
is a sufficient condition to prove the convergence of the damped COS
method.

\begin{assumption}\label{A1} 	
We assume that $\delta_{g}\cap\delta_{w}^\infty\neq\emptyset$ and $\boldsymbol{\alpha}\in\delta_{g}\cap\delta_{w}^\infty$.
\end{assumption}

\begin{assumption}\label{A2} 	
We assume that $\delta_{g}\cap\delta_{w}\neq\emptyset$ and $\boldsymbol{\alpha}\in\delta_{g}\cap\delta_{w}$.
\end{assumption}

Assume Assumption \ref{A1} or \ref{A2} holds. It holds that $wg\in\mathcal{L}^{1}$
since $w$ is bounded and $g$ is a density. For a scaling factor
$\lambda>0$ and a shift parameter $\boldsymbol{\mu}\in\mathbb{R}^{d}$,
we define the \emph{damped} \emph{density} by 
\begin{equation}
f(\boldsymbol{x})=\lambda e^{\boldsymbol{\alpha}\cdot(\boldsymbol{x}+\boldsymbol{\mu})}g(\boldsymbol{x}+\boldsymbol{\mu}),\quad\boldsymbol{x}\in\mathbb{R}^{d}\label{eq:dampedf}
\end{equation}
and the\emph{ damped} \emph{function of interest} by
\begin{equation}
v(\boldsymbol{x})=\frac{1}{\lambda}e^{-\boldsymbol{\alpha}\cdot(\boldsymbol{x}+\boldsymbol{\mu})}w(\boldsymbol{x}+\boldsymbol{\mu}),\quad\boldsymbol{x}\in\mathbb{R}^{d}.\label{eq:dampedv}
\end{equation}
By definition, it follows that
\begin{equation}
\int_{\mathbb{R}^{d}}w(\boldsymbol{x})g(\boldsymbol{x})d\boldsymbol{x}=\int_{\mathbb{R}^{d}}v(\boldsymbol{x})f(\boldsymbol{x})d\boldsymbol{x}.\label{eq:wg=00003Dvf}
\end{equation}

\begin{rem}
\label{rem:fHat_vhat}By Proposition \ref{prop:centr} it holds that
$\widehat{f}(\boldsymbol{u})=\lambda e^{-i\boldsymbol{u}\cdot\boldsymbol{\mu}}\widehat{g}(\boldsymbol{u}-i\boldsymbol{\alpha})$,
i.e., $\widehat{f}$ is given in closed form if $\widehat{g}$ is
given in closed form. The Fourier transform of $v$ is given by $\widehat{v}(\boldsymbol{u})=\lambda^{-1}e^{-i\boldsymbol{u}\cdot\boldsymbol{\mu}}\widehat{w}(\boldsymbol{u}+i\boldsymbol{\alpha})$.
Hence, a closed form expression for $\widehat{w}$ is sufficient to
obtain a closed form expression for $\widehat{v}$. 
\end{rem}

\begin{rem}
Under Assumption \ref{A1} \emph{or} Assumption \ref{A2} it holds
that $vf\in\mathcal{L}^{1}$, $f\in\mathcal{L}^{1}\cap\mathcal{L}^{2}\cap\mathcal{L}^{e}$
and $v\in\mathcal{L}^{\infty}$. Under Assumption \ref{A2} it holds
\emph{additionally} that $v\in\mathcal{L}^{1}\cap\mathcal{L}^{2}\cap\mathcal{L}^{e}$.
\end{rem}

\begin{rem}
On the choice of $\lambda$: In some applications it might be useful
to know that $g$ and $f$ are densities. Thanks to Proposition \ref{prop:centr},
$f$ is a density if $g$ is a density and if we choose $\lambda=\big(\widehat{g}(-i\boldsymbol{\alpha})\big)^{-1}$.
If $\boldsymbol{\alpha}=\boldsymbol{0}$ it holds that $\lambda=1$.
\end{rem}

\begin{rem}
\label{rem:On-the-choice}On the choice of $\boldsymbol{\mu}$: to
apply the COS method, we need to truncate the density $\boldsymbol{x}\mapsto\lambda e^{\boldsymbol{\alpha}\cdot\boldsymbol{x}}g(\boldsymbol{x})$,
i.e., we need to find some interval $[\boldsymbol{a},\boldsymbol{b}]$
such that $\lambda e^{\boldsymbol{\alpha}\cdot\boldsymbol{x}}g(\boldsymbol{x})\approx\lambda e^{\boldsymbol{\alpha}\cdot\boldsymbol{x}}g(\boldsymbol{x})1_{[\boldsymbol{a},\boldsymbol{b}]}(\boldsymbol{x})$.
Assume that such an interval is given. Let $\boldsymbol{\mu}:=\frac{\boldsymbol{b}-\boldsymbol{a}}{2}+\boldsymbol{a}$
and $\boldsymbol{L}:=\frac{\boldsymbol{b}-\boldsymbol{a}}{2}$. Define
$f$ as in Equation (\ref{eq:dampedf}). We then have $f\approx f1_{[-\boldsymbol{L},\boldsymbol{L}]}.$
Therefore, we can always work with $f$ instead of $g$ and a symmetric
interval $[-\boldsymbol{L},\boldsymbol{L}]$ to keep the notation
simple. We recall ideas from the literature, how such interval $[\boldsymbol{a},\boldsymbol{b}]$
can be defined: \citet[Eq. (49)]{fang2009novel} and \citet[Eq. (5.8)]{ruijter2012two}
suggest setting $[a_{h},b_{h}]:=\left[\int_{\mathbb{R}^{d}}\lambda e^{\boldsymbol{\alpha}\cdot\boldsymbol{x}}g(\boldsymbol{x})x_{h}d\boldsymbol{x}\pm L_{h}\right]$
for $L_{h}$ large enough, $h=1,...,d$. We follow this idea and set
\[
\mu_{h}:=\frac{b_{h}-a_{h}}{2}+a_{h}=\int_{\mathbb{R}^{d}}\lambda e^{\boldsymbol{\alpha}\cdot\boldsymbol{x}}g(\boldsymbol{x})x_{h}d\boldsymbol{x},\quad h=1,...,d.
\]
The vector $\boldsymbol{\mu}$ is equal to the first moments of the
marginal densities of $\boldsymbol{x}\mapsto\lambda e^{\boldsymbol{\alpha}\cdot\boldsymbol{x}}g(\boldsymbol{x})$.
Proposition \ref{prop:centr} explains how to obtain $\boldsymbol{\mu}$
from $\widehat{g}$.
\end{rem}

\begin{rem}
\label{rem:theta}In some applications, $\widehat{g}$ is not given
explicitly but needs to be approximated numerically; e.g., in \cite{duffie2003affine},
$\widehat{g}$ is the solution to some ordinary differential equation
(ODE), which itself needs to be solved numerically before solving
the integral (\ref{eq:int}). We denote the approximation of $\widehat{g}$
by $\phi$. The corresponding approximation of $\widehat{f}$ is denoted
by $\vartheta(\boldsymbol{u}):=\lambda e^{-i\boldsymbol{u}\cdot\boldsymbol{\mu}}\phi(\boldsymbol{u}-i\boldsymbol{\alpha})$,
$\boldsymbol{u}\in\mathbb{R}^{d}$. We assume that $\phi(\boldsymbol{u}-i\boldsymbol{\alpha})$
exists for $\boldsymbol{u}\in\mathbb{R}^{d}$ to ensure that $\vartheta$
is well defined. We have to evaluate $\widehat{f}$ at finitely many
points to apply the COS method. So for any fixed $\boldsymbol{u}\in\mathbb{R}^{d}$
and any $\delta>0$ we assume that we can choose $\vartheta$ such
that $|\widehat{f}(\boldsymbol{u})-\vartheta(\boldsymbol{u})|<\delta$.
For example, if $\vartheta$ is obtained by the Euler method, the
global error between $\widehat{f}(\boldsymbol{u})$ and $\vartheta(\boldsymbol{u})$
decreases linearly in the step size for discretization of the ODE,
see Section \ref{subsec:Uncertainty-of-the} for details.
\end{rem}

From now on, we will work with $\widehat{f}$ and $v$ instead of
$\widehat{g}$ and $w$. A numerical approximation of $\widehat{f}$
is denoted by a function $\vartheta:\mathbb{R}^{d}\to\mathbb{C}$.
If there is no numerical uncertainty on $\widehat{f}$, replace $\vartheta$
by $\widehat{f}$ in the rest of the article. 

\begin{table}[H]
\begin{centering}
\begin{tabular}{|c|l|}
\hline 
 & Aim: Approximate $\int_{\mathbb{R}^{d}}w(\boldsymbol{x})g(\boldsymbol{x})d\boldsymbol{x}$
by $\sideset{}{'}\sum\tilde{c}_{\boldsymbol{k}}v_{\boldsymbol{k}}$
or by $\sideset{}{'}\sum\tilde{c}_{\boldsymbol{k}}\tilde{v}_{\boldsymbol{k}}$.\tabularnewline
\hline 
\hline 
$w:\mathbb{R}^{d}\to\mathbb{R}$ & Function of interest\tabularnewline
\hline 
$g:\mathbb{R}^{d}\to\mathbb{R}$ & Density function\tabularnewline
\hline 
$\boldsymbol{\alpha},\boldsymbol{\mu}\in\mathbb{R}^{d}$, $\lambda\in\mathbb{R}_{+}^{d}$ & Damping factor, shift parameter and scaling factor\tabularnewline
\hline 
$v:\mathbb{R}^{d}\to\mathbb{R}$ & Damped function of interest\tabularnewline
\hline 
$f:\mathbb{R}^{d}\to\mathbb{R}$ & Damped density function\tabularnewline
\hline 
$\widehat{w},\widehat{g},\widehat{v},\widehat{f}:\mathbb{R}^{d}\to\mathbb{\mathbb{C}}$ & Fourier transforms of $w$, $g$, $v$ and $f$\tabularnewline
\hline 
$\phi:\mathbb{R}^{d}\to\mathbb{\mathbb{C}}$ & Numerical approximation of $\widehat{g}$\tabularnewline
\hline 
$\vartheta:\mathbb{R}^{d}\to\mathbb{\mathbb{C}}$ & Numerical approximation of $\widehat{f}$\tabularnewline
\hline 
$e_{\boldsymbol{k}}:\mathbb{R}^{d}\to\mathbb{R}$ & Cosine basis functions, see Eq. (\ref{eq:ek})\tabularnewline
\hline 
\textbf{$\boldsymbol{L},\boldsymbol{M}\in\mathbb{R}_{+}^{d}$} & Truncation ranges for $f$ and $v$\tabularnewline
\hline 
$\boldsymbol{N}\in\mathbb{N}^{d}$ & Number of terms\tabularnewline
\hline 
$\boldsymbol{0}\leq\boldsymbol{k}\leq\boldsymbol{N}$ & Index\tabularnewline
\hline 
$a_{\boldsymbol{k}}\in\mathbb{R}$ & Classical Fourier-cosine coefficients of $f1_{[-\boldsymbol{L},\boldsymbol{L}]}$\tabularnewline
\hline 
$c_{\boldsymbol{k}}\in\mathbb{R}$ & Approximation of $a_{\boldsymbol{k}}$ through $\widehat{f}$, see
Eq. (\ref{eq:ck})\tabularnewline
\hline 
$\tilde{c}_{\boldsymbol{k}}\in\mathbb{R}$ & Approximation of $c_{\boldsymbol{k}}$ through $\vartheta$, see Eq.
(\ref{eq:ck_tilde})\tabularnewline
\hline 
$v_{\boldsymbol{k}}\in\mathbb{R}$ & (Scaled) Fourier-cosine coefficients of $v1_{[-\boldsymbol{M},\boldsymbol{M}]}$,
see Eq. (\ref{eq:vk})\tabularnewline
\hline 
$v_{\boldsymbol{k}}^{\text{Num}}\in\mathbb{R}$ & Numerical approximation of $v_{\boldsymbol{k}}$\tabularnewline
\hline 
$\tilde{v}_{\boldsymbol{k}}\in\mathbb{R}$ & Approximation of $v_{\boldsymbol{k}}$ through $\widehat{v}$, see
Eq. (\ref{eq:vkTilde}) \tabularnewline
\hline 
$\mathcal{S}$ & Index set, see Eq. (\ref{eq:setS})\tabularnewline
\hline 
$\mathcal{U}(\boldsymbol{L},\boldsymbol{N})$ & Set where $\widehat{f}$ needs to be evaluated, see Eq. (\ref{eq:U}).\tabularnewline
\hline 
\end{tabular}
\par\end{centering}
\caption{\protect\label{tab:Overview-multidimensional-COS}Overview of the
multidimensional COS method.}
\end{table}

We describe the COS method in detail in order to approximate the right-hand
side of Equation (\ref{eq:wg=00003Dvf}). Table \ref{tab:Overview-multidimensional-COS}
provides an overview. Let $\boldsymbol{M}\in\mathbb{R}_{+}^{d}$ be
large enough so that
\begin{equation}
\int_{\mathbb{R}^{d}}v(\boldsymbol{x})f(\boldsymbol{x})d\boldsymbol{x}\approx\int_{[-\boldsymbol{M},\boldsymbol{M}]}v(\boldsymbol{x})f(\boldsymbol{x})d\boldsymbol{x}.\label{eq:approx_integrate_M}
\end{equation}
Under Assumption \ref{A1} or \ref{A2}, such $\boldsymbol{M}$ exists
since then $vf\in\mathcal{L}^{1}$. Let $\boldsymbol{L}\geq\boldsymbol{M}$.
We truncate $f$ on $[-\boldsymbol{L},\boldsymbol{L}]$ and approximate
the truncated damped density using a Fourier-cosine series. We intuitively
have that
\begin{equation}
f\approx f1_{[-\boldsymbol{L},\boldsymbol{L}]}\approx\sideset{}{'}\sum_{\boldsymbol{\boldsymbol{0}}\leq\boldsymbol{k}\leq\boldsymbol{N}}a_{\boldsymbol{k}}e_{\boldsymbol{k}}1_{[-\boldsymbol{L},\boldsymbol{L}]}\approx\sideset{}{'}\sum_{\boldsymbol{\boldsymbol{0}}\leq\boldsymbol{k}\leq\boldsymbol{N}}c_{\boldsymbol{k}}e_{\boldsymbol{k}}1_{[-\boldsymbol{L},\boldsymbol{L}]}\approx\sideset{}{'}\sum_{\boldsymbol{\boldsymbol{0}}\leq\boldsymbol{k}\leq\boldsymbol{N}}\tilde{c}_{\boldsymbol{k}}e_{\boldsymbol{k}}1_{[-\boldsymbol{L},\boldsymbol{L}]},\label{eq:f_approx}
\end{equation}
where the basis functions are defined by
\begin{equation}
e_{\boldsymbol{k}}(\boldsymbol{x})=\prod_{h=1}^{d}\cos\left(k_{h}\pi\frac{x_{h}+L_{h}}{2L_{h}}\right),\quad\boldsymbol{x}\in\mathbb{R}^{d},\quad\boldsymbol{k}\in\mathbb{\mathbb{N}}_{0}^{d}\label{eq:ek}
\end{equation}
and the classical Fourier-cosine coefficients (see \citet[Part III, Section 9]{pivato2010linear})
of $f1_{[-\boldsymbol{L},\boldsymbol{L}]}$ are given for $\boldsymbol{k}\in\mathbb{\mathbb{N}}_{0}^{d}$
by

\begin{align}
a_{\boldsymbol{k}} & =\frac{1}{\prod_{h=1}^{d}L_{h}}\int_{[-\boldsymbol{L},\boldsymbol{L}]}f(\boldsymbol{x})e_{\boldsymbol{k}}(\boldsymbol{x})d\boldsymbol{x}\nonumber \\
 & \approx\frac{1}{\prod_{h=1}^{d}L_{h}}\int_{\mathbb{R}^{d}}f(\boldsymbol{x})e_{\boldsymbol{k}}(\boldsymbol{x})d\boldsymbol{x}\label{eq:integral_ck}\\
 & =\frac{1}{2^{d-1}\prod_{h=1}^{d}L_{h}}\sum_{\boldsymbol{s}\in\mathcal{S}}\Re\left\{ \widehat{f}\left(\frac{\pi}{2}\frac{\boldsymbol{sk}}{\boldsymbol{L}}\right)\exp\left(i\frac{\pi}{2}\boldsymbol{s}\cdot\boldsymbol{k}\right)\right\} =:c_{\boldsymbol{k}}\label{eq:ck}\\
 & \approx\frac{1}{2^{d-1}\prod_{h=1}^{d}L_{h}}\sum_{\boldsymbol{s}\in\mathcal{S}}\Re\left\{ \vartheta\left(\frac{\pi}{2}\frac{\boldsymbol{sk}}{\boldsymbol{L}}\right)\exp\left(i\frac{\pi}{2}\boldsymbol{s}\cdot\boldsymbol{k}\right)\right\} =:\tilde{c}_{\boldsymbol{k}},\label{eq:ck_tilde}
\end{align}
where
\begin{equation}
\mathcal{S}:=\big\{\boldsymbol{s}\in\mathbb{Z}^{d}:s_{1}=1,s_{h}\in\{-1,1\},\quad h=2,...,d\big\}.\label{eq:setS}
\end{equation}
The key insight of the COS method is the fact that the integral at
the right-hand side of Equation (\ref{eq:integral_ck}) can be solved
explicitly\footnote{We use $\prod_{h=1}^{d}\cos\theta_{h}=\frac{1}{2^{d-1}}\sum_{\boldsymbol{s}\in\mathcal{S}}\cos\left(\boldsymbol{s}\cdot\boldsymbol{\theta}\right)$,
$\boldsymbol{\theta}\in\mathbb{R}^{d}$, which follows by mathematical
induction, from the fact that the cosine is an even function and from
the trigonometric identities stated in \citet[Eqs. (4.3.17, 4.3.31)]{abramowitz1972handbook}.}. If $\widehat{f}$ needs to be approximated by some function $\vartheta$,
we use $\tilde{c}_{\boldsymbol{k}}$ instead of $c_{\boldsymbol{k}}$,
compare with Remark \ref{rem:theta}.

The idea of the multidimensional COS method is to approximate $f$
as in (\ref{eq:f_approx}), and hence the right-hand side of Equation
(\ref{eq:approx_integrate_M}) by

\begin{align}
\int_{[-\boldsymbol{M},\boldsymbol{M}]}v(\boldsymbol{x})f(\boldsymbol{x})d\boldsymbol{x} & \approx\sideset{}{'}\sum_{\boldsymbol{\boldsymbol{0}}\leq\boldsymbol{k}\leq\boldsymbol{N}}\tilde{c}_{\boldsymbol{k}}\underbrace{\int_{[-\boldsymbol{M},\boldsymbol{M}]}v(\boldsymbol{x})e_{\boldsymbol{k}}(\boldsymbol{x})d\boldsymbol{x}}_{=:v_{\boldsymbol{k}}}.\label{eq:vk}
\end{align}
We summarize from the literature three approaches to computing $v_{\boldsymbol{k}}$
and propose a new method as well:
\begin{description}
\item [{COS-i}] Classical COS method (see \citet{fang2009novel}): it is
assumed that $\boldsymbol{\alpha}=\boldsymbol{0}$, $v_{\boldsymbol{k}}$
can be obtained explicitly and that Assumption \ref{A1} holds. Examples,
where $v_{\boldsymbol{k}}$ are given analytically, include, in one
dimension, plain vanilla put and call options and digital options,
see \citet{fang2009novel} and, in two dimensions, geometric basket
options, call-on-maximum options and put-on-minimum options, see \cite{ruijter2012two}.
In general dimensions, the coefficients $v_{\boldsymbol{k}}$ of a
CDF are given in closed form, see Example \ref{exa:CDF}. 
\item [{COS-ii}] Classical COS method approximating $v_{\boldsymbol{k}}$
numerically (see \cite{ruijter2012two}): it is assumed that $\boldsymbol{\alpha}=\boldsymbol{0}$
and that Assumption \ref{A1} holds. This method has been applied
in cases in which the coefficients $v_{\boldsymbol{k}}$ cannot be
obtained explicitly, e.g. in the case of arithmetic basket options,
call-on-minimum and spread options. The coefficients $v_{\boldsymbol{k}}$
are then approximated numerically by the discrete cosine transform
or some quadrature rule. However, solving the integral at the right-hand
side of (\ref{eq:vk}) numerically for each $\boldsymbol{k}$ is expensive
and slows down the COS method significantly. Let us denote the numerical
approximation of $v_{\boldsymbol{k}}$ by $v_{\boldsymbol{k}}^{\text{Num}}$.
\item [{COS-iii}] Classical COS method with damping (see \citet{wang2017pricing}):
it is assumed that $\boldsymbol{\alpha}\neq\boldsymbol{0}$ and $v_{\boldsymbol{k}}$
can be obtained explicitly and that Assumption \ref{A1} holds. \citet{wang2017pricing}
introduces damping (in one dimension) to make the COS method numerically
more stable for unbounded functions of interest. We propose to approximate
the $L_{1}$-norm by this method in Example \ref{exa:L1-norm}. However,
the coefficients $v_{\boldsymbol{k}}$ (with or without damping) are
not given in closed form for many functions of interest, e.g., arithmetic
basket options, call-on-minimum and spread options.
\item [{COS-iv}] Damped COS method (new to the literature): it is assumed
that $\boldsymbol{\alpha}\neq\boldsymbol{0}$, $\widehat{v}$ is given
explicitly and Assumption \ref{A2} holds. The coefficients $v_{\boldsymbol{k}}$
are approximated by $\tilde{v}_{\boldsymbol{k}}$, which are defined
explicitly in terms of $\widehat{v}$ in Equation (\ref{eq:vkTilde}).
This method is much faster than the classical COS method described
in COS-ii since the coefficients $\tilde{v}_{\boldsymbol{k}}$ are
given explicitly while $v_{\boldsymbol{k}}^{\text{Num}}$ are obtained
numerically.
\end{description}
The novel damped COS-iv method is defined as follows: assume that
Assumption \ref{A2} holds, which implies that $v\in\mathcal{L}^{1}.$
Hence, the integral over the finite interval $[-\boldsymbol{M},\boldsymbol{M}]$
at the right-hand side of (\ref{eq:vk}) can be approximated by an
integral over $\mathbb{R}^{d}$ and we approximate the right-hand
side of Equation (\ref{eq:approx_integrate_M}) by
\begin{equation}
\int_{[-\boldsymbol{M},\boldsymbol{M}]}v(\boldsymbol{x})f(\boldsymbol{x})d\boldsymbol{x}\approx\sideset{}{'}\sum_{\boldsymbol{\boldsymbol{0}}\leq\boldsymbol{k}\leq\boldsymbol{N}}\tilde{c}_{\boldsymbol{k}}\underbrace{\int_{\mathbb{R}^{d}}v(\boldsymbol{x})e_{\boldsymbol{k}}(\boldsymbol{x})d\boldsymbol{x}}_{=:\tilde{v}_{\boldsymbol{k}}}.\label{eq:vk_tilde}
\end{equation}
This works if $\boldsymbol{M}$ is large enough. Similar to the solution
presented in Equation (\ref{eq:ck}), the coefficients $\tilde{v}_{\boldsymbol{k}}$
are given analytically:
\begin{equation}
\tilde{v}_{\boldsymbol{k}}=\frac{1}{2^{d-1}}\sum_{\boldsymbol{s}\in\mathcal{S}}\Re\left\{ \widehat{v}\left(\frac{\pi}{2}\frac{\boldsymbol{sk}}{\boldsymbol{L}}\right)\exp\left(i\frac{\pi}{2}\boldsymbol{s}\cdot\boldsymbol{k}\right)\right\} ,\label{eq:vkTilde}
\end{equation}
provided $\widehat{v}$ is known. As shown in Section \ref{sec:Functions-of-interest},
there are many functions of interest such that $\widehat{v}$ is given
explicitly. 
\begin{rem}
In the special case that $\widehat{f}$ only takes real values, the
computational cost of the classical COS-i and the damped COS-iv methods
can be reduced by (about) a factor of one half, since $c_{\boldsymbol{k}}=0$
if $\sum_{h=1}^{d}k_{h}$ is odd. This applies for example in the
case of the multivariate normal distribution, see Example \ref{exa:(Normal-distribution).-Let}.
\end{rem}

The following theorem shows that multivariate densities can be approximated
by a Fourier-cosine series. The theorem also includes numerical uncertainty
on the Fourier transform $\widehat{f}$. Note that the classical COS-i
and the damped COS-iv methods need to evaluate $\widehat{f}$ only
at finitely many points, see Equation (\ref{eq:ck}), given by
\begin{equation}
\mathcal{U}(\boldsymbol{L},\boldsymbol{N}):=\left\{ \frac{\pi}{2}\frac{\boldsymbol{sk}}{\boldsymbol{L}}:\boldsymbol{0}\leq\boldsymbol{k}\leq\boldsymbol{N},\,\boldsymbol{s}\in\mathcal{S}\right\} .\label{eq:U}
\end{equation}

\begin{thm}
\label{thm:approx=000020f}\emph{(COS approximation of $f$).} Assume
that Assumption \ref{A1} or \ref{A2} holds. Let $\vartheta:\mathbb{R}^{d}\to\mathbb{C}$
and define $\tilde{c}_{\boldsymbol{k}}$ as in Equation (\ref{eq:ck_tilde}).
For any $\varepsilon>0$ there exists a $\boldsymbol{L}\in\mathbb{R}_{+}^{d}$,
a $\boldsymbol{N}\in\mathbb{N}^{d}$ and a $\delta>0$ such that $\max_{\boldsymbol{u}\in\mathcal{U}(\boldsymbol{L},\boldsymbol{N})}\left|\widehat{f}(\boldsymbol{u})-\vartheta(\boldsymbol{u})\right|<\delta$
implies
\[
\left\Vert f-\sideset{}{'}\sum_{\boldsymbol{\boldsymbol{0}}\leq\boldsymbol{k}\leq\boldsymbol{N}}\tilde{c}_{\boldsymbol{k}}e_{\boldsymbol{k}}1_{[-\boldsymbol{L},\boldsymbol{L}]}\right\Vert _{2}<\varepsilon.
\]
Note that $\boldsymbol{N}$ depends on $\boldsymbol{L}$ and that
$\delta$ depends on both $\boldsymbol{L}$ and $\boldsymbol{N}$.
\end{thm}

\begin{proof}
Define $e_{\boldsymbol{k}}^{\boldsymbol{L}}=e_{\boldsymbol{k}}1_{[-\boldsymbol{L},\boldsymbol{L}]}$.
It holds for $\boldsymbol{l},\boldsymbol{k}\in\mathbb{N}_{0}^{d}$
that
\begin{align}
\left\langle e_{\boldsymbol{k}}^{\boldsymbol{L}},e_{\boldsymbol{l}}^{\boldsymbol{L}}\right\rangle  & =\begin{cases}
2^{\Lambda(\boldsymbol{k})}\prod_{h=1}^{d}L_{h} & ,\boldsymbol{k}=\boldsymbol{l},\\
0 & ,\text{otherwise,}
\end{cases}\label{eq:ekel}
\end{align}
where $\Lambda$ is defined in Section \ref{sec:Notation}. For any
$\boldsymbol{L}\in\mathbb{R}_{+}^{d}$ and $\boldsymbol{N}\in\mathbb{N}^{d}$,
it holds that
\begin{align*}
\left\Vert f-\sideset{}{'}\sum_{\boldsymbol{\boldsymbol{0}}\leq\boldsymbol{k}\leq\boldsymbol{N}}\tilde{c}_{\boldsymbol{k}}e_{\boldsymbol{k}}^{\boldsymbol{L}}\right\Vert _{2}\leq & \underbrace{\left\Vert f-f1_{[-\boldsymbol{L},\boldsymbol{L}]}\right\Vert _{2}}_{=:A_{1}(\boldsymbol{L})}+\underbrace{\left\Vert f1_{[-\boldsymbol{L},\boldsymbol{L}]}-\sideset{}{'}\sum_{\boldsymbol{\boldsymbol{0}}\leq\boldsymbol{k}\leq\boldsymbol{N}}a_{\boldsymbol{k}}e_{\boldsymbol{k}}^{\boldsymbol{L}}\right\Vert _{2}}_{=:A_{2}(\boldsymbol{L},\boldsymbol{N})}\\
 & +\underbrace{\left\Vert \sideset{}{'}\sum_{\boldsymbol{\boldsymbol{0}}\leq\boldsymbol{k}\leq\boldsymbol{N}}(a_{\boldsymbol{k}}-c_{\boldsymbol{k}})e_{\boldsymbol{k}}^{\boldsymbol{L}}\right\Vert _{2}}_{=:A_{3}(\boldsymbol{L},\boldsymbol{N})}+\underbrace{\left\Vert \sideset{}{'}\sum_{\boldsymbol{\boldsymbol{0}}\leq\boldsymbol{k}\leq\boldsymbol{N}}(c_{\boldsymbol{k}}-\tilde{c}_{\boldsymbol{k}})e_{\boldsymbol{k}}^{\boldsymbol{L}}\right\Vert _{2}}_{=:A_{4}(\boldsymbol{L},\boldsymbol{N})}.
\end{align*}
Further,
\begin{align*}
A_{3}(\boldsymbol{L},\boldsymbol{N})^{2} & =\sum_{\boldsymbol{\boldsymbol{0}}\leq\boldsymbol{k}\leq\boldsymbol{N}}\sum_{\boldsymbol{\boldsymbol{0}}\leq\boldsymbol{l}\leq\boldsymbol{N}}\frac{1}{2^{\Lambda(\boldsymbol{k})+\Lambda(\boldsymbol{l})}}(a_{\boldsymbol{k}}-c_{\boldsymbol{k}})(a_{\boldsymbol{l}}-c_{\boldsymbol{l}})\left\langle e_{\boldsymbol{k}}^{\boldsymbol{L}},e_{\boldsymbol{l}}^{\boldsymbol{L}}\right\rangle \leq\sideset{}{'}\sum_{\boldsymbol{k}\in\mathbb{N}_{0}^{d}}\prod_{h=1}^{d}\{L_{h}\}\left|a_{\boldsymbol{k}}-c_{\boldsymbol{k}}\right|^{2}=B_{f}(\boldsymbol{L}),
\end{align*}
see Definition \ref{def:COS_admissible}. For $\varepsilon>0$, choose
$\boldsymbol{L}\in\mathbb{R}_{+}^{d}$ such that $A_{1}(\boldsymbol{L})<\frac{\varepsilon}{4}$
and $B_{f}(\boldsymbol{L})<\big(\frac{\varepsilon}{4}\big){}^{2}$.
Hence, $A_{3}(\boldsymbol{L},\boldsymbol{N})<\frac{\varepsilon}{4}$.
Then choose $\boldsymbol{N}\in\mathbb{N}^{d}$ such that $A_{2}(\boldsymbol{L},\boldsymbol{N})<\frac{\varepsilon}{4}$.
Such $\boldsymbol{N}$ exists by classical Fourier analysis, see \citet[Theorem 9B.1]{pivato2010linear}.
Let $\Upsilon:=\max_{\boldsymbol{u}\in\mathcal{U}(\boldsymbol{L},\boldsymbol{N})}\left|\widehat{f}(\boldsymbol{u})-\vartheta(\boldsymbol{u})\right|$.
By the definition of $c_{\boldsymbol{k}}$ and $\tilde{c}_{\boldsymbol{k}}$,
see Equation (\ref{eq:ck}), it follows that
\begin{align*}
|c_{\boldsymbol{k}}-\tilde{c}_{\boldsymbol{k}}| & \leq\frac{1}{2^{d-1}\prod_{h=1}^{d}L_{h}}\sum_{\boldsymbol{s}\in\mathcal{S}}\left|\widehat{f}\left(\frac{\pi}{2}\frac{\boldsymbol{sk}}{\boldsymbol{L}}\right)-\vartheta\left(\frac{\pi}{2}\frac{\boldsymbol{sk}}{\boldsymbol{L}}\right)\right|\leq\frac{\Upsilon}{\prod_{h=1}^{d}L_{h}}.
\end{align*}
Similarly to the analysis of $A_{3}$, we have
\begin{align}
A_{4}(\boldsymbol{L},\boldsymbol{N}) & ^{2}\leq\sum_{\boldsymbol{\boldsymbol{0}}\leq\boldsymbol{k}\leq\boldsymbol{N}}\prod_{h=1}^{d}\{L_{h}\}\left|c_{\boldsymbol{k}}-\tilde{c}_{\boldsymbol{k}}\right|^{2}\leq\frac{\Upsilon^{2}}{\prod_{h=1}^{d}L_{h}}\prod_{h=1}^{d}\{N_{h}+1\}.\label{eq:A4}
\end{align}
Choose $\delta=\frac{\varepsilon}{4}\sqrt{\prod_{h=1}^{d}L_{h}}\left(\prod_{h=1}^{d}\{N_{h}+1\}\right)^{-\frac{1}{2}}.$
Then $\Upsilon<\delta$ implies $A_{4}(\boldsymbol{L},\boldsymbol{N})<\frac{\varepsilon}{4}$,
which concludes the proof.
\end{proof}
The next two corollaries provide sufficient conditions to ensure that
the classical COS-i and COS-iii methods and the damped COS-iv methods
approximate the integral (\ref{eq:wg=00003Dvf}) within a predefined
error tolerance $\varepsilon>0$, including numerical uncertainty
on $\widehat{f}$. 
\begin{cor}
\label{cor:generalCase}\emph{(Convergence of the classical COS-i
and COS-iii methods).} Assume that Assumption \ref{A1} holds. The
integral of the product of $f$ and $v$ can be approximated by a
finite sum as follows: Let $\varepsilon>0$. Let $\boldsymbol{M}\in\mathbb{R}_{+}^{d}$
and $\xi>0$ such that
\begin{equation}
\int_{\mathbb{R}^{d}\setminus[-\boldsymbol{M},\boldsymbol{M}]}\left|v(\boldsymbol{x})f(\boldsymbol{x})\right|d\boldsymbol{x}\leq\frac{\varepsilon}{3},\quad\left\Vert v1_{[-\boldsymbol{M},\boldsymbol{M}]}\right\Vert _{2}\leq\xi.\label{eq:Mgeneral}
\end{equation}
Let $\boldsymbol{L}\geq\boldsymbol{M}$ be such that
\begin{equation}
\left\Vert f-f1_{[-\boldsymbol{L},\boldsymbol{L}]}\right\Vert _{2}\leq\frac{\varepsilon}{12\xi}\label{eq:L1general}
\end{equation}
and
\begin{equation}
B_{f}(\boldsymbol{L})\leq\left(\frac{\varepsilon}{12\xi}\right)^{2}.\label{eq:L2general}
\end{equation}
Choose $\boldsymbol{N}\in\mathbb{N}^{d}$ large enough, so that 
\begin{equation}
\left\Vert f1_{[-\boldsymbol{L},\boldsymbol{L}]}-\sideset{}{'}\sum_{\boldsymbol{\boldsymbol{0}}\leq\boldsymbol{k}\leq\boldsymbol{N}}a_{\boldsymbol{k}}e_{\boldsymbol{k}}1_{[-\boldsymbol{L},\boldsymbol{L}]}\right\Vert _{2}\leq\frac{\varepsilon}{12\xi}.\label{eq:Nlargeenough}
\end{equation}
Let $c:=\frac{\varepsilon}{12\xi}\frac{\sqrt{\prod_{h=1}^{d}L_{h}}}{\sqrt{\prod_{h=1}^{d}\{N_{h}+1\}}}$.
Then $c>0$. For $\vartheta:\mathbb{R}^{d}\to\mathbb{C}$ assume
\begin{equation}
\max_{\boldsymbol{u}\in\mathcal{U}(\boldsymbol{L},\boldsymbol{N})}\left|\widehat{f}(\boldsymbol{u})-\vartheta(\boldsymbol{u})\right|\leq c.\label{eq:phi_tilde}
\end{equation}
Then it follows that
\begin{equation}
\left|\int_{\mathbb{R}^{d}}v(\boldsymbol{x})f(\boldsymbol{x})d\boldsymbol{x}-\sideset{}{'}\sum_{\boldsymbol{\boldsymbol{0}}\leq\boldsymbol{k}\leq\boldsymbol{N}}\tilde{c}_{\boldsymbol{k}}v_{\boldsymbol{k}}\right|\leq\varepsilon.\label{eq:Ngeneral}
\end{equation}
\end{cor}

\begin{cor}
\label{cor:(Convergence-damped}\emph{(Convergence of the damped COS-iv
method).} Assume that Assumption \ref{A2} holds. Let $\varepsilon$,
$\xi,$ $\boldsymbol{L}$, $\boldsymbol{M}$, $\boldsymbol{N}$, $c$
and $\vartheta$ be as in Corollary \ref{cor:generalCase}. Let $\eta>0$
such that $\sideset{}{'}\sum_{\boldsymbol{\boldsymbol{0}}\leq\boldsymbol{k}\leq\boldsymbol{N}}|\tilde{c}_{\boldsymbol{k}}|^{2}\leq\eta.$
Assume that
\begin{equation}
\sideset{}{'}\sum_{\boldsymbol{\boldsymbol{0}}\leq\boldsymbol{k}\leq\boldsymbol{N}}|\tilde{v}_{\boldsymbol{k}}-v_{\boldsymbol{k}}|^{2}\leq\frac{\varepsilon^{2}}{9\eta}.\label{eq:v_l_v_l_tilde}
\end{equation}
Then it follows that $\left|\int_{\mathbb{R}^{d}}v(\boldsymbol{x})f(\boldsymbol{x})d\boldsymbol{x}-\sideset{}{'}\sum_{\boldsymbol{\boldsymbol{0}}\leq\boldsymbol{k}\leq\boldsymbol{N}}\tilde{c}_{\boldsymbol{k}}\tilde{v}_{\boldsymbol{k}}\right|\leq\varepsilon$.
\end{cor}

\begin{proof}
We first prove Corollary \ref{cor:(Convergence-damped}. Define $e_{\boldsymbol{k}}^{\boldsymbol{L}}=e_{\boldsymbol{k}}1_{[-\boldsymbol{L},\boldsymbol{L}]}$.
Let $A_{1}(\boldsymbol{L})$, $A_{2}(\boldsymbol{L},\boldsymbol{N})$
and $A_{4}(\boldsymbol{L},\boldsymbol{N})$ be as in the proof of
Theorem \ref{thm:approx=000020f}. By Inequalities (\ref{eq:A4},
\ref{eq:phi_tilde}) it follows that $A_{4}(\boldsymbol{L},\boldsymbol{N})\leq\frac{\varepsilon}{12\xi}$.
Due to $v_{\boldsymbol{k}}=\langle v1_{[-\boldsymbol{M},\boldsymbol{M}]},e_{\boldsymbol{k}}^{\boldsymbol{L}}\rangle$
and applying Theorem \ref{thm:approx=000020f} and the Cauchy-Schwarz
inequality, we have that
\begin{align*}
\Big| & \int_{\mathbb{R}^{d}}v(\boldsymbol{x})f(\boldsymbol{x})d\boldsymbol{x}-\sideset{}{'}\sum_{\boldsymbol{\boldsymbol{0}}\leq\boldsymbol{k}\leq\boldsymbol{N}}\tilde{c}_{\boldsymbol{k}}\tilde{v}_{\boldsymbol{k}}\Big|\\
= & \bigg|\int_{\mathbb{R}^{d}\setminus[-\boldsymbol{M},\boldsymbol{M}]}v(\boldsymbol{x})f(\boldsymbol{x})d\boldsymbol{x}+\langle v1_{[-\boldsymbol{M},\boldsymbol{M}]},f\rangle-\sideset{}{'}\sum_{\boldsymbol{\boldsymbol{0}}\leq\boldsymbol{k}\leq\boldsymbol{N}}\tilde{c}_{\boldsymbol{k}}\langle v1_{[-\boldsymbol{M},\boldsymbol{M}]},e_{\boldsymbol{k}}^{\boldsymbol{L}}\rangle-\sideset{}{'}\sum_{\boldsymbol{\boldsymbol{0}}\leq\boldsymbol{k}\leq\boldsymbol{N}}\tilde{c}_{\boldsymbol{k}}(\tilde{v}_{\boldsymbol{k}}-v_{\boldsymbol{k}})\bigg|\\
\leq & \underbrace{\int_{\mathbb{R}^{d}\setminus[-\boldsymbol{M},\boldsymbol{M}]}|v(\boldsymbol{x})f(\boldsymbol{x})|d\boldsymbol{x}}_{=:D_{1}(\boldsymbol{M})}+\Big|\Big\langle v1_{[-\boldsymbol{M},\boldsymbol{M}]},f-\sideset{}{'}\sum_{\boldsymbol{\boldsymbol{0}}\leq\boldsymbol{k}\leq\boldsymbol{N}}\tilde{c}_{\boldsymbol{k}}e_{\boldsymbol{k}}^{\boldsymbol{L}}\Big\rangle\Big|+\underbrace{\sqrt{\bigg(\,\sideset{}{'}\sum_{\boldsymbol{\boldsymbol{0}}\leq\boldsymbol{k}\leq\boldsymbol{N}}|\tilde{v}_{\boldsymbol{k}}-v_{\boldsymbol{k}}|^{2}\bigg)\bigg(\,\sideset{}{'}\sum_{\boldsymbol{\boldsymbol{0}}\leq\boldsymbol{k}\leq\boldsymbol{N}}|\tilde{c}_{\boldsymbol{k}}|^{2}\bigg)}}_{=:D_{2}(\boldsymbol{N},\boldsymbol{L},\boldsymbol{M})}\\
< & \frac{\varepsilon}{3}+\|v1_{[-\boldsymbol{M},\boldsymbol{M}]}\|_{2}\,\Big\Vert f-\sideset{}{'}\sum_{\boldsymbol{\boldsymbol{0}}\leq\boldsymbol{k}\leq\boldsymbol{N}}\tilde{c}_{\boldsymbol{k}}e_{\boldsymbol{k}}^{\boldsymbol{L}}\Big\Vert_{2}+D_{2}(\boldsymbol{N},\boldsymbol{L},\boldsymbol{M})\\
< & \frac{\varepsilon}{3}+\xi\left(A_{1}(\boldsymbol{L})+A_{2}(\boldsymbol{L},\boldsymbol{N})+\sqrt{B_{f}(\boldsymbol{L})}+A_{4}(\boldsymbol{L},\boldsymbol{N})\right)+D_{2}(\boldsymbol{N},\boldsymbol{L},\boldsymbol{M})\\
\leq & \frac{\varepsilon}{3}+\xi\left(\frac{\varepsilon}{12\xi}+\frac{\varepsilon}{12\xi}+\frac{\varepsilon}{12\xi}+\frac{\varepsilon}{12\xi}\right)+\frac{\varepsilon}{3}=\varepsilon.
\end{align*}
To prove Corollary \ref{cor:generalCase}, replace $\tilde{v}_{\boldsymbol{k}}$
by $v_{\boldsymbol{k}}$ in the inequalities above. Then, $D_{2}(\boldsymbol{N},\boldsymbol{L},\boldsymbol{M})=0$,
and the assertion in Corollary \ref{cor:generalCase} will follow.
\end{proof}
\begin{rem}
From the proof of Corollary 3.8, it can seen that the Inequalities
(\ref{eq:Mgeneral} - \ref{eq:phi_tilde}) could be strengthened replacing
$\frac{\varepsilon}{3}$ with $\frac{\varepsilon}{2}$ and $\frac{\varepsilon}{12\xi}$
with $\frac{\varepsilon}{8\xi}$. The convergence of the classical
COS-ii method approximating $v_{\boldsymbol{k}}$ numerically by $v_{\boldsymbol{k}}^{\text{Num}}$
follows as in Corollary \ref{cor:(Convergence-damped} replacing $\tilde{v}_{\boldsymbol{k}}$
by $v_{\boldsymbol{k}}^{\text{Num}}$. 
\end{rem}

\begin{rem}
Sufficient conditions such that Inequality (\ref{eq:v_l_v_l_tilde})
holds are provided in Corollary \ref{cor:vk_tilde_M_N}.
\end{rem}

\begin{rem}
We do not need all conditions in Assumption \ref{A1} to prove Theorem
\ref{thm:approx=000020f} and Corollary \ref{cor:generalCase}: $f$
does not need to decay exponentially and $v$ does not need to be
bounded; it is sufficient that $v$ is locally in $\mathcal{L}^{1}\cap\mathcal{L}^{2}$.
\end{rem}

\begin{rem}
\label{rem:c}Inequality (\ref{eq:phi_tilde}) tells us how close
the approximation $\vartheta$ has to be to $\widehat{f}$ to ensure
that the overall error of the COS-i, COS-iii or COS-iv method is utmost
$\varepsilon$. In Section \ref{subsec:Uncertainty-of-the}, we approximate
$\widehat{f}$ by numerically solving an ODE, and Inequality (\ref{eq:phi_tilde})
helps us to choose the number of steps for the Euler method.
\end{rem}

\begin{rem}
\label{rem:M>L}$\boldsymbol{M}$ and $\boldsymbol{L}$ in the Corollaries
\ref{cor:generalCase} and \ref{cor:(Convergence-damped} play different
roles: $\boldsymbol{M}$ depends mainly on $v$ and $\boldsymbol{L}$
depends only on $f$. For example, if $v$ is bounded by a very small
bound, the conditions in (\ref{eq:Mgeneral}) are satisfied for $\boldsymbol{M}$
close to zero. However, to satisfy Inequalities (\ref{eq:L1general},
\ref{eq:L2general}), it might be necessary to choose $\boldsymbol{L}$
much larger than $\boldsymbol{M}$.
\end{rem}

The next theorem and two corollaries provide an explicit formula for
the truncation range and an implicit formula for the number of terms
for the classical COS-i and COS-iii methods and the damped COS-iv
method.
\begin{thm}
\emph{\label{thm:(Multidimensional-COS-method}(Classical COS-i and
COS-iii methods: Find $\boldsymbol{M}$ and $\boldsymbol{L}$). }Assume
that Assumption \ref{A1} holds. Let $n\geq2$ be an even number and
assume that the moments of $f$ of $n^{th}-$order exist, i.e.,
\begin{equation}
m_{h}(n):=\int_{\mathbb{R}^{d}}x_{h}^{n}f(\boldsymbol{x})d\boldsymbol{x}=i^{-n}\left.\frac{\partial^{n}}{\partial u_{h}^{n}}\widehat{f}(\boldsymbol{u})\right|_{\boldsymbol{u}=\boldsymbol{0}}\in(0,\infty),\quad h=1,...,d.\label{eq:moments}
\end{equation}
Let $\varepsilon>0$ be small enough. Define
\begin{equation}
M_{h}:=\left(\frac{3d\left\Vert v\right\Vert _{\infty}}{\varepsilon}m_{h}(n)\right)^{\frac{1}{n}},\quad h=1,...,d,\label{eq:m-1-1}
\end{equation}
and $\boldsymbol{L}=\boldsymbol{M}=(M_{1},...,M_{d})\in\mathbb{R}_{+}^{d}$.
There is a $\boldsymbol{N}\in\mathbb{N}_{0}^{d}$ such that
\begin{equation}
\left|\int_{\mathbb{R}^{d}}v(\boldsymbol{x})f(\boldsymbol{x})d\boldsymbol{x}-\sideset{}{'}\sum_{\boldsymbol{\boldsymbol{0}}\leq\boldsymbol{k}\leq\boldsymbol{N}}c_{\boldsymbol{k}}v_{\boldsymbol{k}}\right|\leq\varepsilon.\label{eq:endRes}
\end{equation}
\end{thm}

\begin{cor}
\label{cor:vk_tilde_M_N}\emph{(Damped COS-iv method: Find $\boldsymbol{M}$
and $\boldsymbol{L}$).} Let \emph{$\boldsymbol{M}$ and $\boldsymbol{L}$}
be as in Theorem \ref{thm:(Multidimensional-COS-method}. Assume that
Assumption \ref{A2} holds. There is a $\boldsymbol{N}\in\mathbb{N}_{0}^{d}$
such that
\begin{equation}
\left|\int_{\mathbb{R}^{d}}v(\boldsymbol{x})f(\boldsymbol{x})d\boldsymbol{x}-\sideset{}{'}\sum_{\boldsymbol{\boldsymbol{0}}\leq\boldsymbol{k}\leq\boldsymbol{N}}c_{\boldsymbol{k}}\tilde{v}_{\boldsymbol{k}}\right|\leq\varepsilon.\label{eq:endRes-1}
\end{equation}
\end{cor}

\begin{cor}
\label{cor:NrTerms}\emph{(Classical COS-i, COS-iii and damped COS-iv
methods: Find $\boldsymbol{N}$). }Let \emph{$\boldsymbol{M}$ and
$\boldsymbol{L}$} be as in Theorem \ref{thm:(Multidimensional-COS-method}.
Let $\xi>0$ such that $\left\Vert v1_{[-\boldsymbol{M},\boldsymbol{M}]}\right\Vert _{2}\leq\xi$.
Suppose for some $\boldsymbol{N}\in\mathbb{N}_{0}^{d}$ it holds that
\begin{equation}
\left|(2\pi)^{-d}\int_{\mathbb{R}^{d}}|\widehat{f}(\boldsymbol{u})|^{2}d\boldsymbol{u}-\prod_{h=1}^{d}L_{h}\sideset{}{'}\sum_{\boldsymbol{\boldsymbol{0}}\leq\boldsymbol{k}\leq\boldsymbol{N}}|c_{\boldsymbol{k}}|^{2}\right|\leq\frac{\varepsilon^{2}}{162\xi^{2}}.\label{eq:f_l-akek}
\end{equation}
If Assumption \ref{A1} is satisfied then Inequality (\ref{eq:endRes})
holds and if Assumption \ref{A2} is satisfied then Inequality (\ref{eq:endRes-1})
holds.
\end{cor}

\begin{proof}
\emph{We first prove Theorem \ref{thm:(Multidimensional-COS-method}}:
Assumption \ref{A1} implies that $f\geq0$, $f$ decays exponentially
and $v$ is bounded. Equation (\ref{eq:moments}) follows by \citet[Thm 25.2]{bauer1996probability}.
For $h\in\{1,...,d\}$ let $\pi_{h}:\mathbb{R}^{d}\to\mathbb{R}$,
$\boldsymbol{x}\mapsto x_{h}$. Let $\lambda_{\text{Lebesgue}}^{d}$
be the Lebesgue measure on $\mathbb{R}^{d}$ and define the finite
and positive measure $\zeta:=f\lambda_{\text{Lebesgue}}^{d}$. By
Markov's inequality, it follows that
\[
\int_{\mathbb{R}^{d}\setminus[-\boldsymbol{M},\boldsymbol{M}]}\left|v(\boldsymbol{x})f(\boldsymbol{x})\right|d\boldsymbol{x}\leq\left\Vert v\right\Vert _{\infty}\sum_{h=1}^{d}\zeta\left(\left\{ \boldsymbol{x}\in\mathbb{R}^{d}:|\pi_{h}(\boldsymbol{x})|\geq M_{h}\right\} \right)\leq\left\Vert v\right\Vert _{\infty}\sum_{h=1}^{d}\frac{m_{h}(n)}{M_{h}^{n}}=\frac{\varepsilon}{3}.
\]
The last equality follows by the definition of $\boldsymbol{M}$.
Define $\xi:=\left\Vert v\right\Vert _{\infty}\sqrt{2^{d}\prod_{h=1}^{d}M_{h}}$.
It holds that $\left\Vert v1_{[-\boldsymbol{M},\boldsymbol{M}]}\right\Vert _{2}\leq\xi.$
Hence, the inequalities in (\ref{eq:Mgeneral}) are satisfied. Next,
we use the following auxiliary result: \foreignlanguage{american}{Let
$s\geq0$, $a>0$ and $n\in\mathbb{N}_{0}$ and $d\in\mathbb{N}$.
Then it holds by mathematical induction over $n$ and \citet[Theorem 8.11 ]{amann2009analysis}
that
\begin{align}
\int_{\{\boldsymbol{x}\in\mathbb{R}^{d}:|\boldsymbol{x}|>s\}}e^{-a|\boldsymbol{x}|}|\boldsymbol{x}|^{n}d\boldsymbol{x} & =\frac{d\pi^{\frac{d}{2}}}{\Gamma\left(1+\frac{d}{2}\right)}e^{-as}\frac{(n+d-1)!}{a^{n+d}}\sum_{k=0}^{n+d-1}\frac{(as)^{k}}{k!}.\label{eq:amann}
\end{align}
}For $\varepsilon$ small enough, $\boldsymbol{L}$ is large enough.
Using that $f$ decays exponentially and applying Equation (\ref{eq:amann}),
we obtain with $\ell:=\min_{h=1,...,d}L_{h}$ that
\begin{align}
\left\Vert f-f1_{[-\boldsymbol{L},\boldsymbol{L}]}\right\Vert _{2} & \leq C_{1}\sqrt{\int_{\{\boldsymbol{x}\in\mathbb{R}^{d}:|\boldsymbol{x}|>\ell\}}e^{-2C_{2}|\boldsymbol{x}|}d\boldsymbol{x}}\leq\frac{\varepsilon}{12\xi}.\label{eq:||f-f_l}
\end{align}
The last inequality holds true if $\varepsilon$ is small enough because,
thanks to Equality (\ref{eq:amann}), the term in the middle of (\ref{eq:||f-f_l})
decreases exponentially in $\varepsilon$, while the term at the right-hand
side of (\ref{eq:||f-f_l}) goes to zero like $\varepsilon^{1+\frac{d}{2n}}$
for $\varepsilon\searrow0$. Hence, Inequality (\ref{eq:L1general})
holds. By Inequality (\ref{eq:B(L)_2}) it holds that $B_{f}(\boldsymbol{L})\leq\varepsilon^{2}(12\xi)^{-2}$
if $\varepsilon$ is small enough because $B_{f}(\boldsymbol{L})$
decreases exponentially in $\varepsilon$: to see this, use Inequality
(\ref{eq:||f-f_l}) and observe that the term $\int_{\mathbb{R}^{d}\setminus[-\boldsymbol{L},\boldsymbol{L}]}\left|\boldsymbol{x}\right|^{2d}|f(\boldsymbol{x})|^{2}d\boldsymbol{x}$
converges exponentially thanks to Inequality (\ref{eq:amann}). Hence,
Inequality (\ref{eq:L2general}) holds. By classical Fourier analysis,
there is a $\boldsymbol{N}\in\mathbb{N}_{0}^{d}$ such that Inequality
(\ref{eq:Nlargeenough}) is satisfied. By assumption we have $c_{\boldsymbol{k}}=\tilde{c}_{\boldsymbol{k}}$.
Inequality (\ref{eq:phi_tilde}) holds trivially. Apply Corollary
\ref{cor:generalCase} to finish the proof of Theorem \ref{thm:(Multidimensional-COS-method}.

\emph{We prove Corollary \ref{cor:vk_tilde_M_N}}: We have to show
that Inequality (\ref{eq:v_l_v_l_tilde}) holds to prove Corollary
\ref{cor:vk_tilde_M_N}. Let $G(\boldsymbol{L})$ be as in Equality
(\ref{eq:eta_l}). Observe $G(\boldsymbol{L})\to0$, $\min_{h}L_{h}\to\infty$
because $f$ is COS-admissible by Proposition \ref{prop:COSadmissible}.
There is $\boldsymbol{P}\in\mathbb{R}_{+}^{d}$ and a $\delta>0$
such that $G(\boldsymbol{L})\leq\delta$ for all $\boldsymbol{L}\geq\boldsymbol{P}$.
By Inequality (\ref{eq:bound_ck2}), it follows for all $\boldsymbol{N}\in\mathbb{N}^{d}$
and all $\boldsymbol{L}\geq\boldsymbol{P}$ that
\begin{equation}
\sideset{}{'}\sum_{\boldsymbol{\boldsymbol{0}}\leq\boldsymbol{k}\leq\boldsymbol{N}}|c_{\boldsymbol{k}}|^{2}\leq\frac{\int_{\mathbb{R}^{d}}|f(\boldsymbol{x})|^{2}d\boldsymbol{x}+\delta}{\prod_{h=1}^{d}P_{h}}=:\eta<\infty.\label{eq:eta}
\end{equation}
It follows by Proposition \ref{prop:COSadmissible} for all $\boldsymbol{N}\in\mathbb{N}^{d}$
that
\begin{align}
\sideset{}{'}\sum_{\boldsymbol{\boldsymbol{0}}\leq\boldsymbol{k}\leq\boldsymbol{N}}|\tilde{v}_{\boldsymbol{k}}-v_{\boldsymbol{k}}|^{2} & \leq\sideset{}{'}\sum_{\boldsymbol{k}\in\mathbb{N}_{0}^{d}}\left|\int_{\mathbb{R}^{d}\setminus[-\boldsymbol{M},\boldsymbol{M}]}v(\boldsymbol{x})e_{\boldsymbol{k}}(\boldsymbol{x})d\boldsymbol{x}\right|^{2}\leq\prod_{h=1}^{d}M_{h}B_{v}(\boldsymbol{M})\leq\frac{\varepsilon^{2}}{9\eta}\label{eq:v_tile-vk}
\end{align}
the last inequality holds true if $\varepsilon$ is small enough because
the term $M_{h}B_{v}(\boldsymbol{M})$ decreases exponentially in
$\varepsilon$ since $v$ decays exponentially, while the right-hand
side of (\ref{eq:v_tile-vk}) goes to zero like $\varepsilon^{2}$
for $\varepsilon\searrow0$. 

\emph{We prove Corollary \ref{cor:NrTerms}}: Let $G(\boldsymbol{L})$
be defined as in Equation (\ref{eq:eta_l}). By Lemma \ref{lem:fMinusAkEk}
and the Plancherel theorem, it follows that
\begin{align}
\Vert f1_{[-\boldsymbol{L},\boldsymbol{L}]}-\sideset{}{'}\sum_{\boldsymbol{\boldsymbol{0}}\leq\boldsymbol{k}\leq\boldsymbol{N}}a_{\boldsymbol{k}}e_{\boldsymbol{k}}1_{[-\boldsymbol{L},\boldsymbol{L}]}\Vert_{2}^{2}\leq & \left|(2\pi)^{-d}\int_{\mathbb{R}^{d}}|\widehat{f}(\boldsymbol{u})|^{2}d\boldsymbol{u}-\prod_{h=1}^{d}L_{h}\sideset{}{'}\sum_{\boldsymbol{\boldsymbol{0}}\leq\boldsymbol{k}\leq\boldsymbol{N}}|c_{\boldsymbol{k}}|^{2}\right|+G(\boldsymbol{L})\nonumber \\
\leq & \frac{\varepsilon^{2}}{162\left\Vert v1_{[-\boldsymbol{M},\boldsymbol{M}]}\right\Vert _{2}^{2}}+\frac{\varepsilon^{2}}{162\left\Vert v1_{[-\boldsymbol{M},\boldsymbol{M}]}\right\Vert _{2}^{2}}=\left(\frac{\varepsilon}{9\left\Vert v1_{[-\boldsymbol{M},\boldsymbol{M}]}\right\Vert _{2}}\right)^{2}.\label{eq:fakekPlan}
\end{align}
The last inequality holds because for $\varepsilon>0$ small enough,
$\boldsymbol{L}$ is large enough so that $G(\boldsymbol{L})\leq\frac{\varepsilon^{2}}{162\left\Vert v1_{[-\boldsymbol{M},\boldsymbol{M}]}\right\Vert _{2}^{2}}$
since $G(\boldsymbol{L})$ decreases exponentially. Note that we may
replace the term $\frac{\varepsilon}{12\xi}$ in Inequalities (\ref{eq:L1general},
\ref{eq:L2general}, \ref{eq:Nlargeenough}) in Corollary \ref{cor:generalCase}
by $\frac{\varepsilon}{9\left\Vert v1_{[-\boldsymbol{M},\boldsymbol{M}]}\right\Vert _{2}}$,
since $c_{\boldsymbol{k}}=\tilde{c}_{\boldsymbol{k}}$. Apply Inequality
(\ref{eq:fakekPlan}) to conclude.
\end{proof}
Assume the density $f$ in Corollary \ref{cor:NrTerms} is the density
of a Lévy process at a particular time point and $\widehat{f}$ is
real. The next Proposition \ref{prop:real} shows that the term $(2\pi)^{-d}\int_{\mathbb{R}^{d}}|\widehat{f}(\boldsymbol{u})|^{2}d\boldsymbol{u}$
is then given in closed form if $f$ is known. The densities of many
Lévy processes are given explicitly or in terms of specialized functions,
e.g., for the Meixner process, the Normal Inverse Gaussian process,
the Variance Gamma process and the Generalized Hyperbolic process,
see \cite{barndorff1997normal,madan1998variance,schoutens2003levy}
and references therein. For example, the multivariate Brownian motion
at time $T>0$ has a multivariate normal density and the corresponding
Fourier transform is real (with or without damping), see Example \ref{exa:(Normal-distribution).-Let}.
\begin{prop}
\label{prop:real}Let $(\boldsymbol{X}_{t})_{t\geq0}$ be a $d$-dimensional
Lévy process. Assume that the characteristic function $\widehat{f}_{\boldsymbol{X}_{t}}$
of $\boldsymbol{X}_{t}$ is real for all $t>0$ and that $\boldsymbol{X}_{t}$
has a density, denoted by $f_{\boldsymbol{X}_{t}}$. Let $T>0$. Then
\[
(2\pi)^{-d}\int_{\mathbb{R}^{d}}|\widehat{f}_{\boldsymbol{X}_{T}}(\boldsymbol{u})|^{2}d\boldsymbol{u}=f_{\boldsymbol{X}_{2T}}(\boldsymbol{0}).
\]
\end{prop}

\begin{proof}
Using that $\boldsymbol{X}$ has independent and stationary increments
and that $\widehat{f}_{X_{T}}$ is real, it follows that
\begin{align*}
(2\pi)^{-d}\int_{\mathbb{R}^{d}}|\widehat{f}_{\boldsymbol{X}_{T}}(\boldsymbol{u})|^{2}d\boldsymbol{u}= & (2\pi)^{-d}\int_{\mathbb{R}^{d}}\big(\widehat{f}_{\boldsymbol{X}_{T}}(\boldsymbol{u})\big)^{2}d\boldsymbol{u}=(2\pi)^{-d}\int_{\mathbb{R}^{d}}\widehat{f}_{\boldsymbol{X}_{2T}}(\boldsymbol{u})d\boldsymbol{u}=f_{\boldsymbol{X}_{2T}}(\boldsymbol{0}).
\end{align*}
\end{proof}
\begin{rem}
\label{rem:lessStable}Provided the expression $(2\pi)^{-d}\int_{\mathbb{R}^{d}}|\widehat{f}(\boldsymbol{u})|^{2}d\boldsymbol{u}$
can be obtained precisely, Inequality (\ref{eq:f_l-akek}) makes it
possible to define a stopping criterion for $\boldsymbol{N}$. In
particular, Inequality (\ref{eq:f_l-akek}) enables us to determine
$\boldsymbol{N}$ \emph{while} computing the coefficients $c_{\boldsymbol{k}}$:
incrementally increase $\boldsymbol{N}$ and compute $c_{\boldsymbol{k}}$
and $|c_{\boldsymbol{k}}|^{2}$ simultaneously. Stop when Inequality
(\ref{eq:f_l-akek}) is met. This is explained in more detail in Algorithm
\ref{alg:Application-of-Lemma}. However, since the right-hand side
of Equation (\ref{eq:f_l-akek}) converges to zero at least like $O\left(\varepsilon^{2}\right)$,
rounding off errors makes it difficult to find $\boldsymbol{N}$ by
Inequality (\ref{eq:f_l-akek}) for very small $\varepsilon$. Using
arbitrary-precision arithmetic instead of fixed-precision arithmetic
should overcome this drawback.
\end{rem}

The next theorem implies that classical COS-i and COS-iii and the
damped COS-iv methods converge exponentially if $\widehat{f}$ decays
exponentially, i.e., if Inequality (\ref{eq:f_hat_p}) holds for all
$p>0$. The Cases (a) and (b) in Theorem \ref{thm:order=000020of=000020convergence}
treat the classical COS-i and COS-ii methods and the damped COS-iv
method, respectively. The bound for the order of convergence of the
damped COS-iv method is slightly better.

\begin{thm}
\label{thm:order=000020of=000020convergence}(Order of convergence).
Let $\gamma>0$ and $\beta\in(0,1)$. For $n\in\mathbb{N}$, let \textbf{$\boldsymbol{N}=(n,...,n)\in\mathbb{N}^{d}$}
and $\boldsymbol{M}=\boldsymbol{L}=(\gamma n^{\beta},...,\gamma n^{\beta}).$
Assume for some $p>\frac{d}{2}$ that
\begin{equation}
|\widehat{f}(\boldsymbol{u})|\leq O\left(|\boldsymbol{u}|_{\infty}^{-p}\right),\quad|\boldsymbol{u}|_{\infty}\to\infty.\label{eq:f_hat_p}
\end{equation}
(a) \emph{(Classical COS-i or COS-iii methods).} Assume that Assumption
\ref{A1} holds. Then it follows that
\[
\left|\int_{\mathbb{R}^{d}}v(\boldsymbol{x})f(\boldsymbol{x})d\boldsymbol{x}-\sideset{}{'}\sum_{\boldsymbol{\boldsymbol{0}}\leq\boldsymbol{k}\leq\boldsymbol{N}}c_{\boldsymbol{k}}v_{\boldsymbol{k}}\right|\leq O\bigg(n^{-(1-\beta)p+\frac{d}{2}}\bigg),\quad n\to\infty.
\]
(b) \emph{(Damped COS-iv method)}. Assume that Assumptions \ref{A2}
holds. Then it follows that
\[
\left|\int_{\mathbb{R}^{d}}v(\boldsymbol{x})f(\boldsymbol{x})d\boldsymbol{x}-\sideset{}{'}\sum_{\boldsymbol{\boldsymbol{0}}\leq\boldsymbol{k}\leq\boldsymbol{N}}c_{\boldsymbol{k}}\tilde{v}_{\boldsymbol{k}}\right|\leq O\bigg(n^{-(1-\beta)(p-\frac{d}{2})}\bigg),\quad n\to\infty.
\]
\end{thm}

\begin{proof}
Let $A_{1}(\boldsymbol{L})$, $A_{2}(\boldsymbol{L},\boldsymbol{N})$,
$D_{1}(\boldsymbol{M})$ and $D_{2}(\boldsymbol{N},\boldsymbol{L},\boldsymbol{M})$
be as in the proof of Corollary \ref{cor:generalCase}. Since $v_{\boldsymbol{k}}=\langle v1_{[-\boldsymbol{M},\boldsymbol{M}]},e_{\boldsymbol{k}}^{\boldsymbol{L}}\rangle$
and similarly to the proof of Corollary \ref{cor:generalCase} we
have that
\begin{align}
 & \bigg|\int_{\mathbb{R}^{d}}v(\boldsymbol{x})f(\boldsymbol{x})d\boldsymbol{x}-\sideset{}{'}\sum_{\boldsymbol{\boldsymbol{0}}\leq\boldsymbol{k}\leq\boldsymbol{N}}c_{\boldsymbol{k}}r_{\boldsymbol{k}}\bigg|\nonumber \\
\leq & D_{1}(\boldsymbol{M})+\|v1_{[-\boldsymbol{M},\boldsymbol{M}]}\|_{2}\,\bigg(A_{1}(\boldsymbol{L})+A_{2}(\boldsymbol{L},\boldsymbol{N})+\sqrt{B_{f}(\boldsymbol{L})}\bigg)+D_{2}(\boldsymbol{N},\boldsymbol{L},\boldsymbol{M}),\label{dsd}
\end{align}
where $r_{\boldsymbol{k}}:=v_{\boldsymbol{k}}$ in case a) and $r_{\boldsymbol{k}}:=\tilde{v}_{\boldsymbol{k}}$
in case b). We will analyze the order of convergence of each term
at the right-hand side of Inequality (\ref{dsd}): Since $v$ is bounded
and $f$ decays exponentially, $D_{1}(\boldsymbol{M})$, $A_{1}(\boldsymbol{L})$
and $\sqrt{B_{f}(\boldsymbol{L})}$ decay exponentially, i.e., can
be bounded by $O\big(\exp(-C_{3}n^{\beta})\big),$ $n\to\infty$,
for some $C_{3}$, see proof of Theorem \ref{thm:(Multidimensional-COS-method}.
By Inequality (\ref{eq:eta}), the term $\sideset{}{'}\sum|c_{\boldsymbol{k}}|^{2}$
is bounded. In case a), $D_{2}(\boldsymbol{N},\boldsymbol{L},\boldsymbol{M})=0$.
In case b), $D_{2}(\boldsymbol{N},\boldsymbol{L},\boldsymbol{M})$
decays exponentially, see proof of Corollary \ref{cor:vk_tilde_M_N}.
Last, we treat $A_{2}(\boldsymbol{L},\boldsymbol{N})$. Let $j\in\{1,...,d\}$.
Let $n$ be large enough. Let $\boldsymbol{k}\in\mathbb{N}_{0}^{d}$
such that $k_{j}>n$. By Equation (\ref{eq:ck}), Inequality (\ref{eq:f_hat_p})
and using $L_{h}=\gamma n^{\beta}$, $h=1,...,d$, there is a constant
$a_{1}>0$ so that
\begin{align*}
|c_{\boldsymbol{k}}|^{2}\overset{(\ref{eq:ck})}{\leq}\bigg(\frac{1}{2^{d-1}\prod_{h=1}^{d}L_{h}}\sum_{\boldsymbol{s}\in\mathcal{S}}\bigg|\widehat{f}\left(\frac{\pi}{2}\frac{\boldsymbol{sk}}{\boldsymbol{L}}\right)\bigg|\bigg)^{2}\overset{(\text{\ref{eq:f_hat_p}})}{\leq}a_{1}\left(\frac{1}{n^{d\beta}}\left(\frac{|\boldsymbol{k}|_{\infty}}{n^{\beta}}\right)^{-p}\right)^{2}= & a_{1}n^{2\beta(p-d)}|\boldsymbol{k}|_{\infty}^{-2p}.
\end{align*}
By mathematical induction over $d$ and the applying the integral
test of convergence, one can show that
\begin{align}
\sum_{\boldsymbol{k}\in\mathbb{N}_{0}^{d},k_{j}>n}|\boldsymbol{k}|_{\infty}^{-2p} & \leq\frac{2^{d-1}}{(2p-d)n^{2p-d}}.\label{eq:induction}
\end{align}
It follows by Inequality (\ref{eq:induction}) for some $a_{2}>0$
that
\begin{align}
\prod_{h=1}^{d}L_{h}\sum_{\boldsymbol{k}\in\mathbb{N}_{0}^{d},k_{j}>n}|c_{\boldsymbol{k}}|^{2}\leq & a_{2}n^{-(1-\beta)(2p-d)}.\label{eq:k_j_N_j}
\end{align}
Let \textbf{$G(\boldsymbol{L})$ }be defined as in Equality (\ref{eq:eta_l}).
By Equality (\ref{eq:ekel}), the Cauchy-Schwarz (CS) inequality and
Inequality (\ref{eq:bound_ck2}), we obtain
\begin{align*}
A_{2}(\boldsymbol{L},\boldsymbol{N})^{2}\overset{(\ref{eq:ekel})}{\leq} & \prod_{h=1}^{d}L_{h}\sideset{}{'}\sum_{k_{1}>N_{1}\text{ or}\dots\text{or }k_{d}>N_{d}}|a_{\boldsymbol{k}}+c_{\boldsymbol{k}}-c_{\boldsymbol{k}}|^{2}\\
\overset{(\text{CS})}{\leq} & \prod_{h=1}^{d}L_{h}\sideset{}{'}\sum_{k_{1}>N_{1}\text{ or}\dots\text{or }k_{d}>N_{d}}|c_{\boldsymbol{k}}|^{2}+\prod_{h=1}^{d}L_{h}\sideset{}{'}\sum_{\boldsymbol{k}\in\mathbb{N}_{0}^{d}}|a_{\boldsymbol{k}}-c_{\boldsymbol{k}}|^{2}\\
 & +2\sqrt{\prod_{h=1}^{d}L_{h}\sideset{}{'}\sum_{\boldsymbol{k}\in\mathbb{N}_{0}^{d}}|c_{\boldsymbol{k}}|^{2}\prod_{h=1}^{d}L_{h}\sideset{}{'}\sum_{\boldsymbol{k}\in\mathbb{N}_{0}^{d}}|a_{\boldsymbol{k}}-c_{\boldsymbol{k}}|^{2}}\\
\overset{(\ref{eq:bound_ck2})}{\leq} & \sum_{j=1}^{d}\bigg(\prod_{h=1}^{d}L_{h}\sum_{\boldsymbol{k}\in\mathbb{N}_{0}^{d},k_{j}>n}|c_{\boldsymbol{k}}|^{2}\bigg)+B_{f}(\boldsymbol{L})+2\sqrt{\left(\int_{\mathbb{R}^{d}}|f(\boldsymbol{x})|^{2}d\boldsymbol{x}+G(\boldsymbol{L})\right)B_{f}(\boldsymbol{L})}\\
\overset{(\ref{eq:k_j_N_j})}{\leq} & O\left(n^{-(1-\beta)(2p-d)}\right),\quad n\to\infty,
\end{align*}
since $B_{f}(\boldsymbol{L})$ and $G(\boldsymbol{L})$, converge
exponentially to zero. Since $v$ is bounded, we have that $\|v1_{[-\boldsymbol{M},\boldsymbol{M}]}\|_{2}\leq O\left(n^{\frac{d\beta}{2}}\right)$,
$n\to\infty$. Noting that $\frac{-(1-\beta)(2p-d)+d\beta}{2}=-(1-\beta)p+\frac{d}{2}$,
shows (a). It holds $\|v1_{[-\boldsymbol{M},\boldsymbol{M}]}\|_{2}\leq\|v\|_{2}$
if $v\in\mathcal{L}^{2}$, which implies (b).
\end{proof}

\section{\protect\label{sec:Characteristic-functions}Characteristic functions}

In this section, in Examples \ref{exa:(Normal-distribution).-Let}
and \ref{exa:(Variance-Gamma-distribution).}, we recall the multivariate
normal and the Variance Gamma distributions from the literature. Remark
\ref{rem:In-a-financial} and Examples \ref{rem:BS} and \ref{rem:VG}
provide a financial context.
\begin{example}
\label{exa:(Normal-distribution).-Let}(Multivariate normal distribution).
Let $\boldsymbol{X}$ be a multivariate normal random variable with
location $\boldsymbol{\eta}\in\mathbb{R}^{d}$ and covariance matrix
$\Sigma\in\mathbb{R}^{d\times d}$. The random variable $\boldsymbol{X}$
has characteristic function $\widehat{g}(\boldsymbol{u})=\exp\left(i\boldsymbol{\eta}\cdot\boldsymbol{u}-\frac{1}{2}\boldsymbol{u}\cdot\Sigma\boldsymbol{u}\right)$,
$\boldsymbol{u}\in\mathbb{R}^{d}$, which can be extended to $\mathbb{C}^{d}$,
i.e., $\widehat{g}(\boldsymbol{u}-i\boldsymbol{\alpha})$ exists for
all $\boldsymbol{\alpha}\in\mathbb{R}^{d}$ and for the set $\delta_{g}$,
defined in Equation (\ref{eq:delta_g}), it holds that $\delta_{g}=\mathbb{R}^{d}$.
By Proposition \ref{prop:centr} we set $\lambda=\exp\left(-\boldsymbol{\eta}\cdot\boldsymbol{\alpha}-\frac{1}{2}\boldsymbol{\alpha}\cdot\Sigma\boldsymbol{\alpha}\right)$
and $\boldsymbol{\mu}=\boldsymbol{\eta}+\Sigma\boldsymbol{\alpha}$.
The characteristic function of the damped density $f$, defined in
Equation (\ref{eq:dampedf}), is given by $\widehat{f}(\boldsymbol{u})=\exp\left(-\frac{1}{2}\boldsymbol{u}\cdot\Sigma\boldsymbol{u}\right)$,
which is the characteristic function of a multivariate normal random
variable with location zero and covariance matrix $\Sigma$. A straightforward
computation shows that
\begin{align*}
(2\pi)^{-d}\int_{\mathbb{R}^{d}}|\widehat{f}(\boldsymbol{u})|^{2}d\boldsymbol{u} & =\frac{2^{-d}}{\sqrt{\pi{}^{d}\det(\Sigma)}}.
\end{align*}
\end{example}

\begin{example}
\label{exa:(Variance-Gamma-distribution).}(Variance Gamma distribution).
Let $\boldsymbol{Z}$ be a $d$-dimensional, standard normal random
variable. Let $G$ be a Gamma distributed random variable, independent
of $\boldsymbol{Z}$, with shape $a>0$ and scale $s>0$. Let $\boldsymbol{\eta},\boldsymbol{\theta}\in\mathbb{R}^{d}$
and $\boldsymbol{\sigma}\in\mathbb{R}_{+}^{d}$. Consider $\boldsymbol{X}=\boldsymbol{\eta}+\boldsymbol{\theta}G+\sqrt{G}\boldsymbol{\sigma}\boldsymbol{Z}$.
The distribution of $\boldsymbol{X}$ is denoted by $\text{VG}(a,s,\boldsymbol{\eta},\boldsymbol{\theta},\boldsymbol{\sigma})$.
Throughout the paper we assume that $a>\frac{1}{2}$. Define $\Sigma\in\mathbb{R}^{d\times d}$
such that $\Sigma_{ii}=\sigma_{i}^{2}$ and $\Sigma_{ij}=0$ for $i\neq j$.
Then $\boldsymbol{X}$ has the characteristic function
\[
\widehat{g}(\boldsymbol{u})=\exp\left(i\boldsymbol{\eta}\cdot\boldsymbol{u}\right)\big(1-is\boldsymbol{\theta}\cdot\boldsymbol{u}+\frac{1}{2}s\boldsymbol{u}\cdot\Sigma\boldsymbol{u}\big)^{-a},
\]
see \citet{luciano2006multivariate}. These authors propose the simplifying
assumption that $\Sigma$ is a diagonal matrix to ensure that the
number of parameters of the VG distribution increases linearly with
the dimension $d$. The (extended) Fourier transform $\widehat{g}(\boldsymbol{u}-i\boldsymbol{\alpha})$
exists for all $\boldsymbol{\alpha}\in\mathbb{R}^{d}$ with $\zeta(\boldsymbol{\alpha}):=1-s\boldsymbol{\theta}\cdot\boldsymbol{\alpha}-\frac{1}{2}s\boldsymbol{\alpha}\cdot\Sigma\boldsymbol{\alpha}>0$,
see \citet{bayer2022optimal}. By Proposition \ref{prop:centr}, we
set $\lambda=\exp\left(-\boldsymbol{\eta}\cdot\boldsymbol{\alpha}\right)\big(\zeta(\boldsymbol{\alpha})\big)^{a}$
and $\boldsymbol{\mu}=\boldsymbol{\eta}+as\zeta^{-1}(\boldsymbol{\theta}+\Sigma\boldsymbol{\alpha})$.
The characteristic function of the damped density $f$, defined in
Equation (\ref{eq:dampedf}), is given by 
\begin{align*}
\widehat{f}(\boldsymbol{u})= & \exp\left(-i\frac{as}{\zeta(\boldsymbol{\alpha})}\big(\boldsymbol{\theta}+\Sigma\boldsymbol{\alpha}\big)\cdot\boldsymbol{u}\right)\left(1-i\frac{s}{\zeta(\boldsymbol{\alpha})}\big(\boldsymbol{\theta}+\Sigma\boldsymbol{\alpha}\big)\cdot\boldsymbol{u}+\frac{1}{2}\frac{s}{\zeta(\boldsymbol{\alpha})}\boldsymbol{u}\cdot\Sigma\boldsymbol{u}\right)^{-a},
\end{align*}
which is the characteristic function of a $\text{VG}\left(a,\frac{s}{\zeta(\boldsymbol{\alpha})},-\frac{as}{\zeta(\boldsymbol{\alpha})}\big(\boldsymbol{\theta}+\Sigma\boldsymbol{\alpha}\big),\boldsymbol{\theta}+\Sigma\boldsymbol{\alpha},\boldsymbol{\sigma}\right)$
distributed random variable. Using the fact that $\Sigma$ is a diagonal
matrix, it holds for $|\boldsymbol{u}|_{\infty}$ large enough that
\[
|\widehat{f}(\boldsymbol{u})|^{2}\leq\left(\frac{1}{2}\frac{s}{\zeta(\boldsymbol{\alpha})}\boldsymbol{u}\cdot\Sigma\boldsymbol{u}\right)^{-2a}\leq\left(c_{1}|\boldsymbol{u}|^{2}\right)^{-2a}\leq c_{2}|\boldsymbol{u}|_{\infty}^{-4a}
\]
for suitable constants $c_{1},c_{2}>0$. Since $\widehat{f}$ is bounded
by $1$, this implies that $|\widehat{f}(\boldsymbol{u})|\leq O\left(|\boldsymbol{u}|_{\infty}^{-2a}\right)$
for $|\boldsymbol{u}|_{\infty}\to\infty$. Since $a>\frac{1}{2}$,
$\widehat{f}$ is integrable, which implies that $f$ is continuous
and bounded. In particular, $f\in\mathcal{L}^{1}\cap\mathcal{L}^{2}$.
Since the marginal densities of $f$ are VG distributed and have semi-heavy
tails, see \citet{kuchler2008shapes}, $f$ decays exponentially if
$\zeta(\boldsymbol{\alpha})>0$. For the set $\delta_{g}$, defined
in Equation (\ref{eq:delta_g}), it holds that $\delta_{g}=\{\boldsymbol{\alpha}\in{\mathbb{R}^{d}}:\zeta(\boldsymbol{\alpha})>0\}$.
Note that $\boldsymbol{0}\in\delta_{g}$ and $\delta_{g}$ is open.
\end{example}

\begin{rem}
The Assumption that $\widehat{f}$ is integrable in Example \ref{exa:(Variance-Gamma-distribution).},
i.e., $a>\frac{1}{2}$, is also made in \citet[Remark 2.3 and Assumption A3]{eberlein2010analysis}
to treat discontinuous functions of interest, e.g., to obtain a CDF.
\end{rem}

\begin{rem}
\label{rem:In-a-financial}In a financial context, we model $d$ stock
prices over time by a $d$-dimensional positive semimartingale $(\boldsymbol{S}(t))_{t\geq0}$
on a filtered probability space $(\Omega,\mathcal{F},P,(\mathcal{F}_{t})_{t\geq0})$.
The filtration $(\mathcal{F}_{t})_{t\geq0}$ satisfies the usual conditions
with $\mathcal{F}_{0}=\{\Omega,\emptyset\}$. The logarithmic returns
are defined by $\boldsymbol{X}(t):=\log(\boldsymbol{S}(t))$, $t\geq0$.
There is a bank account paying continuous compound interest $r\in\mathbb{R}$.
There is a European option $w:\mathbb{R}^{d}\to\mathbb{R}$ with maturity
$T>0$ and payoff $w\big(\boldsymbol{X}(T)\big)$ at time $T$. We
denote by $g$ the (risk-neutral) density of $\log(\boldsymbol{S}(T))$.
The time-0 price of the European option is then given by $e^{-rT}\int_{\mathbb{R}^{d}}w(\boldsymbol{x})g(\boldsymbol{x})d\boldsymbol{x}.$
This integral can be approximated by the classical COS-i, COS-ii,
COS-iii or the damped COS-iv method.
\end{rem}

\begin{example}
\label{rem:BS}(BS model). Let $\Sigma\in\mathbb{R}^{d\times d}$
be a symmetric positive definite matrix. For the Black-Scholes (BS)
model, the logarithmic returns $\boldsymbol{X}(T)$ are normally distributed
with location $\boldsymbol{\eta}:=\log(\boldsymbol{S}(0))+(\boldsymbol{r}-\frac{1}{2}\text{diag}(\Sigma))T$
and covariance matrix $T\Sigma$, where $\boldsymbol{r}=(r,...,r)\in\mathbb{R}^{d}$
and $\text{diag}(\Sigma)\in\mathbb{R}^{d}$ denotes the diagonal of
$\Sigma$.
\end{example}

\begin{example}
\label{rem:VG}(VG model). Let $\nu>0$, $\boldsymbol{\sigma}\in\mathbb{R}_{+}^{d}$
and $\boldsymbol{\theta}\in\mathbb{R}^{d}$. In the multivariate Variance
Gamma (VG) model, see \citet{luciano2006multivariate}, the logarithmic
returns $\boldsymbol{X}(T)$ follow a $\text{VG}(\frac{T}{\nu},\nu,\boldsymbol{\eta},\boldsymbol{\theta},\boldsymbol{\sigma})$
distribution, where
\[
\eta_{h}:=\log(S_{h}(0))+\big(r+\frac{1}{\nu}\log\big(1-\frac{1}{2}\sigma_{h}^{2}\nu-\theta_{h}\nu\big)\big)T,\quad h=1,\dots,d.
\]
As in Example \ref{exa:(Variance-Gamma-distribution).}, we assume
$\frac{T}{\nu}>\frac{1}{2}$ to ensure that the density of the VG
distribution is square-integrable. \citet{carr1999option} calibrated
the one-dimensional VG model to real market data and observed an average
value of $\nu=0.16$. \citet{luciano2006multivariate} calibrated
a three-dimensional VG model with stochastic time change to real market
data and obtained a similar value for $\nu$. For such $\nu$, only
options with a maturity greater than about one month satisfy the condition
$\frac{T}{\nu}>\frac{1}{2}$, which limits the applicability of the
COS method in the VG model for options with short maturities.
\end{example}

\section{\protect\label{sec:Functions-of-interest}Functions of interest}
\begin{example}
\label{exa:CDF}(CDF by the classical COS-i method). Let $w(\boldsymbol{x})=1_{(-\boldsymbol{\infty},\boldsymbol{y}]}(\boldsymbol{x})$,
$\boldsymbol{x}\in\mathbb{R}^{d}$ for some $\boldsymbol{y}\in\mathbb{R}^{d}$.
Then $v(\boldsymbol{x})=1_{(-\boldsymbol{\infty},\boldsymbol{y}]}(\boldsymbol{x}+\boldsymbol{\mu})$
for shift parameter $\boldsymbol{\mu}$. Since $w$ is bounded and
zero outside $(-\boldsymbol{\infty},\boldsymbol{y}]$, it holds that
$\delta_{w}^{\infty}=\mathbb{R}_{-}^{d}$, where $\delta_{w}^{\infty}$
is defined in Equation (\ref{eq:delta_w_inf}). The integral in (\ref{eq:int})
is equal to the CDF of the density $g$ evaluated at $\boldsymbol{y}$.
The coefficients $v_{\boldsymbol{k}}$, defined in Equation (\ref{eq:vk}),
can be obtained in closed form: Let $\boldsymbol{M},\boldsymbol{L}\in\mathbb{R}_{+}^{d}$
as in Section \ref{sec:Damped-COS-method}. It holds for $\boldsymbol{k}\in\mathbb{N}_{0}^{d}$
that $v_{\boldsymbol{k}}=0$ if $y_{h}-\mu_{h}<-M_{h}$ for any $h$,
and otherwise
\begin{align}
v_{\boldsymbol{k}} & =\prod_{h=1}^{d}\int_{-M_{h}}^{M_{h}}1_{(-\infty,y_{h}]}(x+\mu_{h})\cos\left(k_{h}\pi\frac{x+L_{h}}{2L_{h}}\right)dx\nonumber \\
 & =\prod_{\underset{k_{h}=0}{h=1}}^{d}\{A_{h}+M_{h}\}\prod_{\underset{k_{h}>0}{h=1}}^{d}\bigg\{\frac{2L_{h}}{\pi k_{h}}\bigg(\sin\big(k_{h}\pi\frac{A_{h}+L_{h}}{2L_{h}}\big)-\sin\big(k_{h}\pi\frac{-M_{h}+L_{h}}{2L_{h}}\big)\bigg)\bigg\},\label{eq:vkCDF}
\end{align}
where $A_{h}:=\min(y_{h}-\mu_{h},M_{h})$. It holds that $\left\Vert v\right\Vert _{\infty}\leq1$
and $\left\Vert v1_{[-\boldsymbol{M},\boldsymbol{M}]}\right\Vert _{2}^{2}\leq2^{d}\prod_{h=1}^{d}M_{h}$.
\end{example}

We also discuss the possibility of approximating the CDF by the damped
COS-iv method in Example \ref{exa:(Digital-cash-or-nothing-put} to
test the damped COS-iv and to be able to compare the damped and the
classical COS-i methods. In real world applications we recommend applying
Example \ref{exa:CDF} instead of Example \ref{exa:(Digital-cash-or-nothing-put}
to approximate a CDF since the classical COS-i method does not require
a damping factor and can therefore be applied more straightforwardly.
\begin{example}
\label{exa:(Digital-cash-or-nothing-put}(CDF by the damped COS-iv
method). Let $w$ be as in Example \ref{exa:CDF}. A simple calculation
shows that the Fourier transform of $w$ exists for $\boldsymbol{z}\in\mathbb{C}^{d}$
such that $\Im\{z_{h}\}<0$, $h=1,...,d$, and is given by $\widehat{w}(\boldsymbol{z})=\prod_{h=1}^{d}\frac{e^{iy_{h}z_{h}}}{iz_{h}}.$
Hence, for $\boldsymbol{\alpha}<\boldsymbol{0}$, it holds that
\[
\int_{\mathbb{R}^{d}}|w(\boldsymbol{x})e^{-\boldsymbol{\alpha}\cdot\boldsymbol{x}}|d\boldsymbol{x}=\int_{\mathbb{R}^{d}}w(\boldsymbol{x})e^{-\boldsymbol{\alpha}\cdot\boldsymbol{x}}d\boldsymbol{x}=\widehat{w}(i\boldsymbol{\alpha})\in\mathbb{R}.
\]
So, the map $\boldsymbol{x}\mapsto w(\boldsymbol{x})e^{-\boldsymbol{\alpha}\cdot\boldsymbol{x}}$
is integrable and bounded and therefore square-integrable. It also
satisfies Inequality (\ref{eq:propBed}) since it has semi-heavy tails.
In conclusion, for the set $\delta_{w}$, defined in Equation (\ref{eq:delta_w}),
it holds that $\delta_{w}=\mathbb{R}_{-}^{d}$. For $\lambda>0$ and
$\boldsymbol{\mu}\in\mathbb{R}^{d}$, let $v$ be as in Equation (\ref{eq:dampedv}).
It holds for $\boldsymbol{\alpha}<\boldsymbol{0}$ that $\left\Vert v\right\Vert _{\infty}\leq\lambda^{-1}e^{-\boldsymbol{\alpha}\cdot\boldsymbol{y}}$
and
\begin{align*}
\left\Vert v\right\Vert _{2}^{2} & =\lambda^{-2}\prod_{h=1}^{d}\frac{\exp\left(-2\alpha_{h}(y_{h})\right)}{-2\alpha_{h}}.
\end{align*}
 
\end{example}

\begin{example}
\label{exa:(Various-other-European}(A cash-or-nothing put option
by the classical COS-i method). The payoff function of a cash-or-nothing
put option is defined by $w(\boldsymbol{x})=1_{[\boldsymbol{0},\boldsymbol{K}]}(e^{\boldsymbol{x}})$,
$\boldsymbol{x}\in\mathbb{R}^{d}$ for some $\boldsymbol{K}\in\mathbb{R}_{+}^{d}$;
for a financial context, see Remark \ref{rem:In-a-financial}. The
option pays $1\$$ at maturity if $\boldsymbol{S}(T)\leq\boldsymbol{K}$
and nothing otherwise. The price of a cash-or-nothing put option can
be approximated by the classical COS-i method as in Example \ref{exa:CDF}
with $\boldsymbol{y}:=\log(\boldsymbol{K})$ since $1_{[\boldsymbol{0},\boldsymbol{K}]}(e^{\boldsymbol{x}})=1_{(-\boldsymbol{\infty},\log(\boldsymbol{K})]}(\boldsymbol{x})$.
In particular, the coefficients $v_{\boldsymbol{k}}$ of the classical
COS-i method for a cash-or-nothing put option are as in Equation (\ref{eq:vkCDF}),
replacing $\boldsymbol{y}$ by $\log(\boldsymbol{K})$. We have that
$\delta_{w}^{\infty}=\delta_{w}=\mathbb{R}_{-}^{d}$. According to
Example \ref{exa:(Digital-cash-or-nothing-put}, it holds that $\widehat{w}(\boldsymbol{z})=\prod_{h=1}^{d}\frac{e^{i\log(K_{h})z_{h}}}{iz_{h}}$
provided that $\Im\{z_{h}\}<0$, $h=1,...,d$.
\end{example}

\begin{example}
\label{exa:L1-norm}($L_{1}$-norm by the classical COS-iii method
with damping). Let $w(\boldsymbol{x})=\sum_{i=1}^{d}|x_{i}|$. The
integral in (\ref{eq:int}) is equal to the expected $L_{1}$-norm
of a random variable with density $g$. Obtaining the $L_{1}$-norm
is a challenging numerical problem even in one dimension, see \cite{von1965convergence, brown1972formulae, barndorff2005absolute}.
Note that
\begin{align}
\int_{\mathbb{R}^{d}}w(\boldsymbol{x})g(\boldsymbol{x})d\boldsymbol{x} & =\sum_{i=1}^{d}\int_{\mathbb{R}}|x|g_{i}(x)dx=\sum_{i=1}^{d}\left(\int_{\mathbb{R}}\max(x,0)g_{i}(x)dx+\int_{\mathbb{R}}\max(-x,0)g_{i}(x)dx\right)\label{eq:L1_damped}
\end{align}
where $g_{i}$ is the $i^{th}$ marginal density of $g$, which has
Fourier transform $u\mapsto\widehat{g}(0,...0,u,0,...0)$. So, the
approximation of the expected $L_{1}$-norm boils down to solving
$d$ times a one-dimensional integration problem using the middle
term in (\ref{eq:L1_damped}). It can be solved by the classical COS-i
method since $v_{k}=\int_{-M}^{M}|x|\cos\left(k\pi\frac{x+L}{2L}\right)dx$,
defined in Equation (\ref{eq:vk}), can be obtained in closed form.
However, since $x\mapsto|x|$ is not bounded, Theorem \ref{thm:(Multidimensional-COS-method}
and Corollary \ref{cor:NrTerms} cannot be applied. Furthermore, the
COS-i method is numerically unstable when applied to unbounded functions
of interest due to significant cancellation errors, see \citet[Remark 5.2]{fang2009novel}.
In order to solve the $d$-dimensional integral, i.e., the term at
the left in (\ref{eq:L1_damped}), we propose applying the one-dimensional
classical COS-iii method with damping exactly $2d$ times to approximate
the term at the right in (\ref{eq:L1_damped}), since the functions
$w_{+}(x):=\max(x,0)e^{-\alpha_{+}x}$ and $w_{-}(x):=\max(-x,0)e^{-\alpha_{-}x}$
are bounded for $\alpha_{+}>0$ and $\alpha_{-}<0$. For both functions,
the coefficients $v_{k}$, defined in Equation (\ref{eq:vk}), are
given in closed form since $\int_{M}^{M}w_{\pm}(x)\cos\left(k\pi\frac{x+L}{2L}\right)dx$
can be computed explicitly. For the sets $\delta_{w_{\pm}}^{\infty}$,
defined in Equation (\ref{eq:delta_w_inf}), it holds that $\delta_{w_{\pm}}^{\infty}=\mathbb{R}_{\pm}$.
\end{example}

\begin{example}
\label{exa:w_hat_basket}(Unweighted arithmetic basket put option
by the damped COS-iv method). An unweighted arithmetic  basket put
option is defined by $w(\boldsymbol{x})=\max(K-\sum_{h=1}^{d}e^{x_{h}},0)$,
$\boldsymbol{x}\in\mathbb{R}^{d}$ for some $K>0$; for a financial
context, see Remark \ref{rem:In-a-financial}. The Fourier transform
of $w$ exists for $\boldsymbol{z}\in\mathbb{C}^{d}$ such that $\Im\{z_{h}\}<0$,
$h=1,...,d$, and is given by
\begin{align}
\widehat{w}(\boldsymbol{z}) & =\int_{\mathbb{R}^{d}}e^{i\boldsymbol{z}\cdot\boldsymbol{x}}w(\boldsymbol{x})d\boldsymbol{x}=\frac{K^{\big(1+i\sum_{h=1}^{d}z_{h}\big)}\prod_{h=1}^{d}\Gamma(iz_{h})}{\Gamma\bigg(i\sum_{h=1}^{d}z_{h}+2\bigg)}.\label{eq:w_fourier_basket}
\end{align}
Equation (\ref{eq:w_fourier_basket}) follows by an elementary substitution\footnote{We thank Friedrich Hubalek from Technische Universität Wien for pointing
this out to us.} from \citet[Eq. (5.14.1)]{olver2010nist} and is also mentioned in
a similar form in \cite{hubalek2003variance}. For the set $\delta_{w}$,
defined in Equation (\ref{eq:delta_w}), it holds that $\delta_{w}=\mathbb{R}_{-}^{d}$.
If $\boldsymbol{\alpha}<\boldsymbol{0}$, it holds that $\left\Vert v\right\Vert _{\infty}\leq\lambda^{-1}K^{1-\sum_{h=1}^{d}\alpha_{h}}$
and, using \citet[Eq. (5.14.1)]{olver2010nist} once more, it follows
that
\[
\left\Vert v1_{[-\boldsymbol{M},\boldsymbol{M}]}\right\Vert _{2}^{2}\leq\left\Vert v\right\Vert _{2}^{2}\leq\frac{K^{2-2\sum_{h=1}^{d}\alpha_{h}}}{\lambda^{2}}\frac{\prod_{h=1}^{d}\Gamma\big(-2\alpha_{h}\big)}{\Gamma\big(1+\sum_{h=1}^{d}(-2\alpha_{h})\big)}.
\]
\end{example}

There are other payoff functions with known Fourier transform which
could be priced using the damped COS-iv method: For put and call options
on the maximum or minimum of $d$ assets, see \cite{eberlein2010analysis};
for spread options, see \cite{hurd2010fourier}.

\section{\protect\label{sec:Numerical}Numerical experiments}

We provide several numerical experiments to solve the integral in
(\ref{eq:int}) using the classical COS-i or the damped COS-iv methods.
Reference values are obtained by a Monte Carlo simulation as follows:
Let $\varepsilon>0$ be some error tolerance. If not stated otherwise,
$\varepsilon$ is interpreted as an absolute error tolerance. Let
$p$ be the probability that the difference between the Monte Carlo
estimator with $U$ runs and the true value of the integral is less
than the error tolerance $\varepsilon$. Using the central limit theorem,
we choose $U$ such that $p=0.99$. The Monte Carlo estimator is denoted
by $\text{MC}(\varepsilon)$.

In all experiments, we set $\boldsymbol{M}$ as in Inequality (\ref{eq:m-1-1})
and we set $\boldsymbol{L}:=\boldsymbol{M}$. We only report $\boldsymbol{L}$.
All experiments are performed on a modern laptop with Intel i7-11850H
processor and 32 GB RAM. A short and pure R implementation of the
classical COS-i and the damped COS-iv methods with focus on readability
can be found in Algorithm \ref{alg:Pure-R-code} in the appendix.
For the numerical experiments, we implemented the COS-i and COS-iv
methods and the Monte Carlo simulation in C++ (this code is written
with focus on performance) using for-loops without parallelization.
The memory requirements are minimal and the complete code can be found
in \url{https://github.com/GeroJunike/COS_Method}.

\subsection{\protect\label{subsec:Uncertainty-of-the}Uncertainty of the Fourier
transform}

We illustrate uncertainty on $\widehat{f}$ as in Corollary \ref{cor:generalCase},
where $\widehat{f}$ is defined by a one-dimensional ODE
. An application involving multidimensional ODEs, so-called generalized
Riccati equations, can be found in \cite{duffie2003affine}. Suppose
that $\widehat{f}$ solves the following ODE:
\begin{equation}
\frac{d}{du}\widehat{f}(u)=-u\widehat{f}(u),\quad u\in\mathbb{R},\quad\widehat{f}(0)=1.\label{eq:ODE}
\end{equation}
The ODE can be solved numerically by the Euler method with $Q\in\mathbb{N}$
steps, where a numerical solution $\vartheta_{Q}:[0,\bar{u}]\to\mathbb{R}$,
with $\bar{u}>0$, is constructed as follows: for the step size $h=\frac{\bar{u}}{Q}$,
set $\vartheta_{Q}(0):=1$, 
\[
\vartheta_{Q}\big((q+1)h\big):=\vartheta_{Q}(qh)+h\left(-qh\vartheta_{Q}(qh)\right),\quad q=0,..,Q-1,
\]
and use linear interpolation in the gaps, see e.g., \citet[Section 2]{griffiths2010numerical}.
The analytic solution of the ODE (\ref{eq:ODE}) is also well known
and is given by $\widehat{f}(u)=e^{-\frac{u^{2}}{2}}$, which is the
characteristic function of a standard normal random variable. We will
compare the numerical solution to the analytical solution. After solving
the ODE (\ref{eq:ODE}) by the Euler method, we use the classical
COS-i method with $\vartheta_{Q}$ instead of $\widehat{f}$ to approximate
$\Phi(-2)=\int_{\infty}^{\infty}v(x)f(x)dx$, where $\Phi$ and $f$
are the CDF and density of a standard normal random variable, respectively,
and $v(x)=1_{(-\infty,-2]}(x)$. To use the classical COS-i method,
set $N=5$ and $L=\pi$. We must approximate $\widehat{f}$ by $\vartheta_{Q}$
at $\mathcal{U}(L,N)=\left\{ 0,\frac{1}{2},1,\frac{3}{2},2,\frac{5}{2}\right\} $,
compare with Equation (\ref{eq:ck_tilde}). So we set $\bar{u}=\frac{5}{2}$.
We use the error tolerance $\varepsilon=0.0006$, which is about twice
as large as the absolute difference between the classical COS-i method
without uncertainty on $\widehat{f}$ and the analytical solution
given by $\Phi(-2)$. Since $v$ is bounded by $1$, we set $\xi$
in (\ref{eq:Mgeneral}) to $\sqrt{2L}$. How can we choose $Q$ so
that the error of the COS-i method using $\vartheta_{Q}$ instead
of $\widehat{f}$ remains below $\varepsilon$? Corollary \ref{cor:generalCase}
gives a theoretical answer: We have to choose $Q$ such that Inequality
(\ref{eq:phi_tilde}) is satisfied, i.e., 

\[
\max_{u\in\mathcal{U}(L,N)}\left|\widehat{f}(u)-\vartheta_{Q}(u)\right|\leq c:=\frac{\varepsilon}{12\xi}\frac{\sqrt{L}}{\sqrt{N+1}}\approx1.4\times10^{-5}.
\]
Note that $\log(c)\approx-11.1$. We observe empirically from Figure
\ref{fig:euler} that for $Q\approx e^{10.5}$ Inequality (\ref{eq:phi_tilde})
is satisfied. We also observe that the COS-i approximation then has
an error of $0.00031$, which is well below the error tolerance $\varepsilon$.
For $Q\to\infty$, we see that the error of the COS-i method using
$\vartheta_{Q}$ converges to the error of the COS-i method using
$\widehat{f}$. However, we also see in this example that $Q\approx e^{6.4}$
is already sufficient to stay below the error tolerance $\varepsilon$.
In conclusion, we confirm empirically that if Inequality (\ref{eq:phi_tilde})
is satisfied, the Euler approximation does not increase the error
of the COS-i method too much. However, the smallest $Q$ that keeps
the error of the COS-i method below $\varepsilon$ is about 60 times
smaller than predicted by Inequality (\ref{eq:phi_tilde}).

\begin{figure}[H]
\centering{}\includegraphics[scale=0.4]{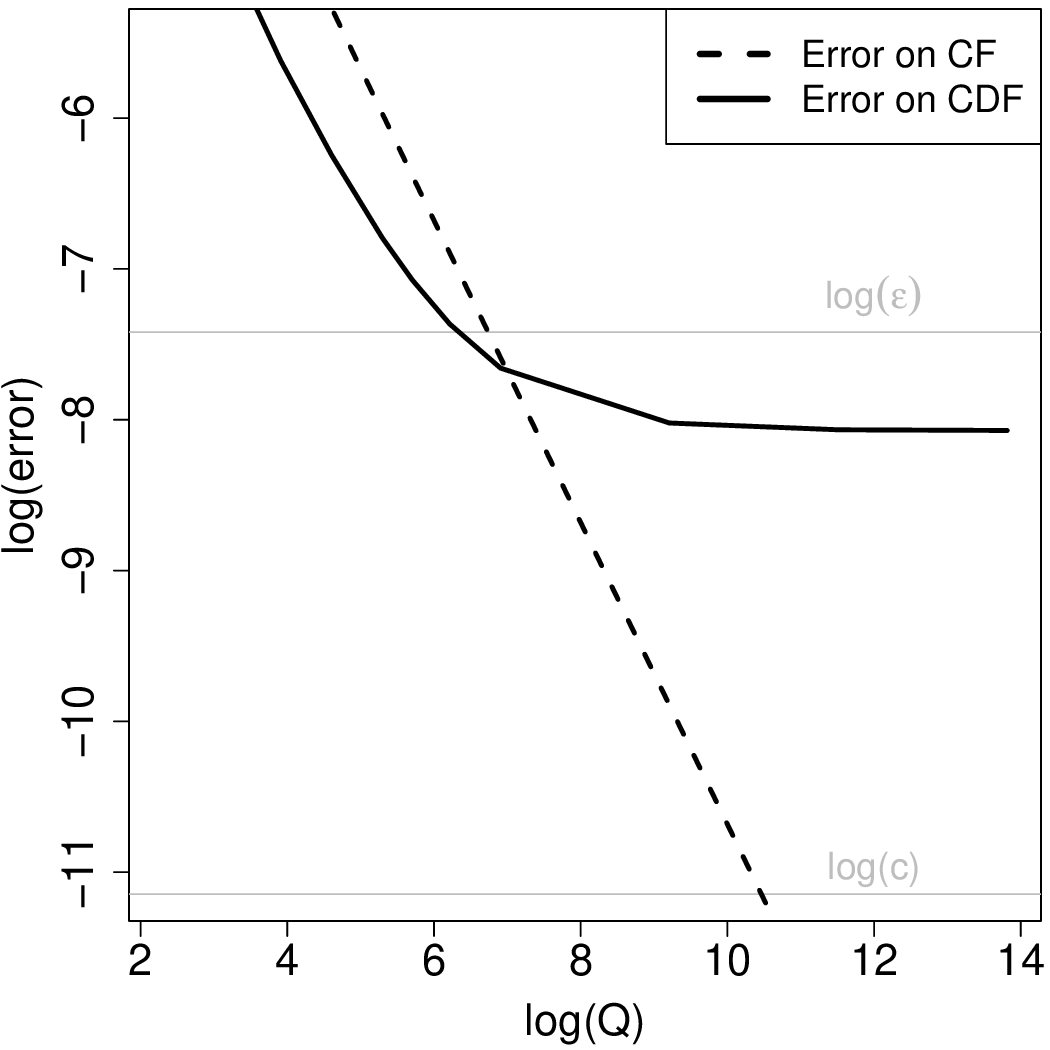}\caption{\protect\label{fig:euler}How the absolute error on the characteristic
function (CF) propagates in the COS-i method to estimate the CDF $\Phi$
of a standard normal random variable at $-2$. The solid line shows
the absolute error between $\Phi(-2)$ and the COS-i method, where
the characteristic function is approximated by the Euler method with
$Q$ steps. The two gray horizontal lines are the logarithmic error
tolerance $\log(\varepsilon)$ and $\log(c)$.}
\end{figure}

\subsection{\protect\label{subsec:On-the-choice}On the choice of the damping
factor and the order of convergence}

We investigate the influence of the damping factor $\boldsymbol{\alpha}$
on the accuracy of the damped COS-iv method to obtain the CDF of a
multivariate normal distribution. According to Example \ref{exa:(Various-other-European},
the CDF can also be interpreted as the price of a cash-or-nothing
put option. Figure \ref{fig:alpha} shows the behavior of the damped
COS-iv method for different damping factors in dimensions $d\in\{2,3,4\}$.
If $\boldsymbol{\alpha}$ is too close to zero, almost no damping
takes place and the difference between $v_{\boldsymbol{k}}$ and $\tilde{v}_{\boldsymbol{k}}$
is large, which implies a relatively high error for the damped COS-iv
method. If $|\boldsymbol{\alpha}|$ is too big, $\left\Vert v\right\Vert _{\infty}$
and $\left\Vert v\right\Vert _{2}$ become very large and the truncation
error increases. 

We observe in Figure \ref{fig:alpha} that a wide range of damping
factors work well in various dimensions. Fixing the number of terms
$\boldsymbol{N}$ and the truncation range $\boldsymbol{L}$, we see
that the classical COS-i method almost always yields lower error than
the damped COS-iv method, except for very specific choices for $\boldsymbol{\alpha}$,
which can only be determined through trial and error and is not known
a priori.

Table \ref{tab:Truncation-ranges-for} shows for the multivariate
normal distribution the dependence of the truncation range on the
damping factor and on the correlation, i.e., on the off-diagonal elements
of the covariation matrix, in $d\in\{2,4\}$ dimensions. The truncation
range does not depend on the off-diagonal elements if there is no
damping. However, in the case of high correlations, large damping
factors in four dimensions greatly increase the truncation range.

We also illustrate the order of convergence of the damped COS-iv method
for an unweighted arithmetic basket put option in the VG model. We
compare three different maturities. In Figure \ref{fig:alpha} we
can see that the theoretical bound from Theorem \ref{thm:order=000020of=000020convergence}
for the order of convergence is sharp and close to the empirical order
of convergence for $T\geq0.5$. The theoretical bound is less sharp
for short maturities like $T=0.1$. 

\begin{figure}[H]
\begin{centering}
\includegraphics[scale=0.4]{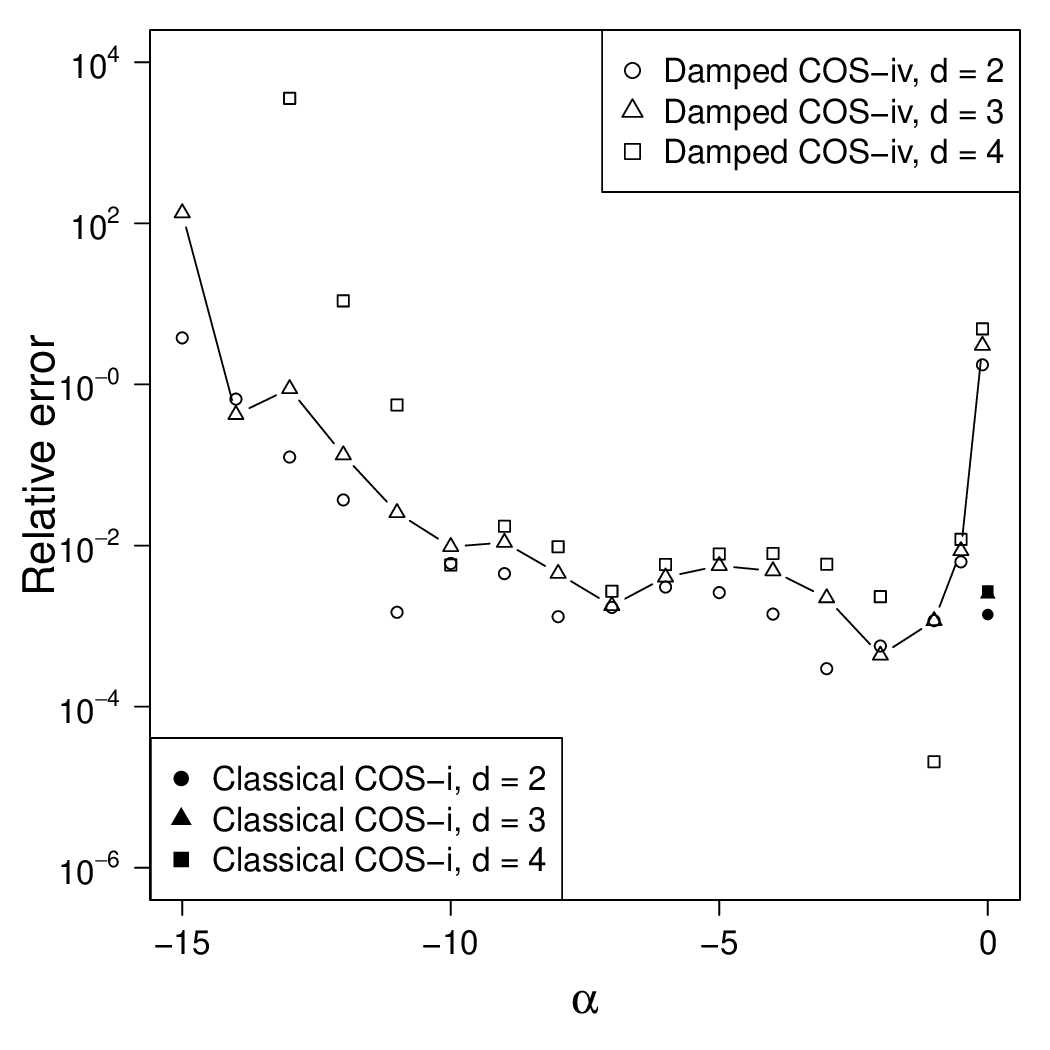} \includegraphics[scale=0.4]{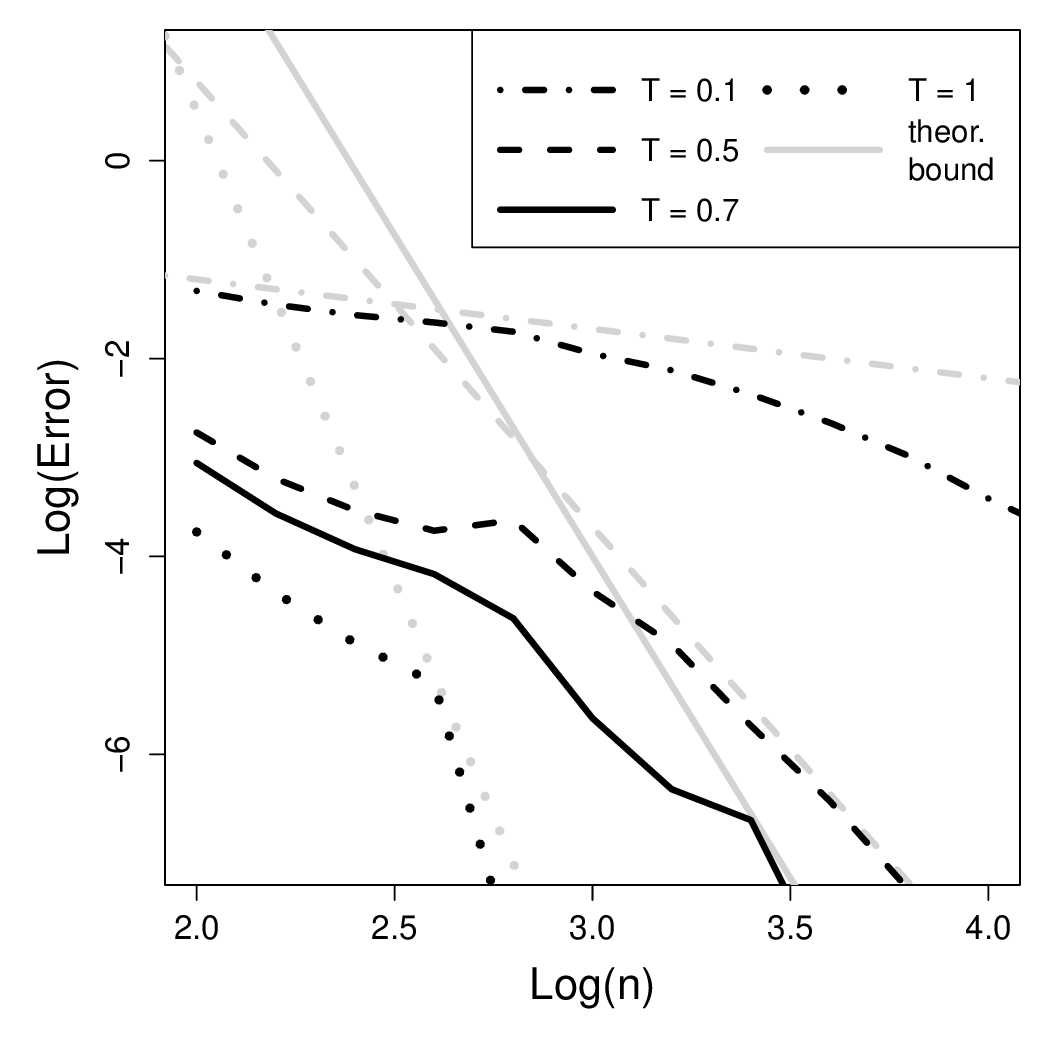}
\par\end{centering}
\caption{\protect\label{fig:alpha}\emph{Left}: Relative error of the CDF of
the multivariate normal distribution with different damping factors,
$\boldsymbol{\eta}=(4.58517,...,4.58517)$, $\boldsymbol{y}=(4.60517,...,4.60517)$
and $\Sigma_{ii}=\sigma^{2}$, $\Sigma_{ij}=0.5\sigma^{2}$, $i\protect\neq j$,
where $\sigma=0.2$. Further, $\boldsymbol{L}=(30\sigma,...,30\sigma)$
and $\boldsymbol{N}=(60,...,60)$. Reference values are obtained by
a Monte Carlo $\text{MC}(10^{-5})$ simulation. \emph{Right}: Logarithmic
error of the price by the damped COS-iv method for the VG model over
the logarithmic number of terms for an unweighted arithmetic basket
put option and $d=2$. We choose $\boldsymbol{N}=(n,n$) and $\boldsymbol{L}=(\gamma n^{\beta},\gamma n^{\beta})$
with $\gamma=\beta=\frac{1}{2}$. We set $\boldsymbol{S}(0)=(50,50)$,
$K=100$,\textbf{ $\boldsymbol{\sigma}=(0.2,0.2)$},\textbf{ $\boldsymbol{\theta}=(-0.03,-0.03)$},\textbf{
$\nu=0.1$}, $r=0$ and $\boldsymbol{\alpha}=(-4,-4)$. The theoretical
bound from Theorem \ref{thm:order=000020of=000020convergence}, i.e.,
a line with slope $-(1-\beta)(p-\frac{d}{2})$, is shown in gray.
For the VG model, we have $p=\frac{2T}{\nu}$. Reference values are
obtained by a Monte Carlo $\text{MC}(10^{-4})$ simulation and are
given by 1.6244, 3.8998, 4.6509 and 5.5951 for $T=0.1$, $T=0.5$,
$T=0.7$ and $T=1$, respectively.}
\end{figure}

\begin{table}[H]
\begin{centering}
\begin{tabular}{|c|c|c|c|c|}
\hline 
 & $\alpha=0$ & $\alpha=-3$ & $\alpha=-7$ & $\alpha=-11$\tabularnewline
\hline 
\hline 
$\rho=0.00$ & 1.42 / 1.54 & 1.50 / 1.74 & 1.87 / 2.70 & 2.74 / ~~5.78\tabularnewline
\hline 
$\rho=0.25$ & 1.42 / 1.54 & 1.52 / 1.86 & 1.99 / 3.90 & 3.19 / ~14.32\tabularnewline
\hline 
$\rho=0.50$ & 1.42 / 1.54 & 1.54 / 1.99 & 2.12 / 5.64 & 3.71 / ~35.49\tabularnewline
\hline 
$\rho=0.75$ & 1.42 / 1.54 & 1.55 / 2.13 & 2.25 / 8.14 & 4.31 / ~87.95\tabularnewline
\hline 
$\rho=0.99$ & 1.42 / 1.54 & 1.57 / 2.27 & 2.39 / 11.6 & 4.99 / 210.1\tabularnewline
\hline 
\end{tabular}
\par\end{centering}
\caption{\protect\label{tab:Truncation-ranges-for}Truncation ranges $\boldsymbol{L}=(L,...,L)$
for error tolerance $\varepsilon=10^{-4}$ for a multivariate normal
distribution evaluated at $\boldsymbol{y}=(4.60517,...,4.60517)$
with location $\boldsymbol{\eta}=(4.58517,...,4.58517)$ and covariance
matrix $\Sigma_{ii}=\sigma^{2}$, $\Sigma_{ij}=\rho\sigma^{2}$, $i\protect\neq j$,
where $\sigma=0.2$. We show $L$ for the dimensions $d=2$ / $d=4$
for different values for $\rho$ and damping factors $\boldsymbol{\alpha}=(\alpha,...,\alpha)$.
The term $L$ is obtained from Equation (\ref{eq:m-1-1}) and a bound
for $|v|_{\infty}$ is taken from Example \ref{exa:(Digital-cash-or-nothing-put}.}
\end{table}

\subsection{\protect\label{subsec:On-the-choiceCDF}On the choice of $\boldsymbol{N}$}

In this section, we approximate the CDF for the VG and the multivariate
normal distributions by the classical COS-i method and the price of
an unweighted arithmetic put option in the BS model by the damped
COS-iv method. 

We approximate the CDF for the $\text{VG}(a,s,\boldsymbol{\eta},\boldsymbol{\theta},\boldsymbol{\sigma})$
distribution in $d=3$ dimensions with parameters $a=10$, $s=0.1$,
$\boldsymbol{\eta}=\boldsymbol{0}$, $\boldsymbol{\theta}=(-0.03,-0.03,-0.03)$
and $\boldsymbol{\sigma}=(0.2,0.2,0.2)$ at 1000 different $\boldsymbol{y}\in\mathbb{R}^{d}$
by the classical COS-i method. These $\boldsymbol{y}$ are randomly
chosen according to the $\text{VG}(a,s,\boldsymbol{\eta},\boldsymbol{\theta},\boldsymbol{\sigma})$
distribution. 

We set the error tolerance to $\varepsilon=10^{-3}$ and obtain $\boldsymbol{L}=(1.2,1.2,1.2)$
using $n=8$ moments by Inequality (\ref{eq:m-1-1}). The number of
terms is obtained by Corollary \ref{cor:NrTerms} and is given by
$\boldsymbol{N}=(21,21,21)$. We solve the integral appearing in Corollary
\ref{cor:NrTerms} numerically. This takes some time, but it only
has to be done once: the formulas for $\boldsymbol{L}$ and $\boldsymbol{N}$
depend only on the upper bound of the function of interest and not
on the concrete choice of $\boldsymbol{y}$. Therefore, we use identical
$\boldsymbol{L}$ and $\boldsymbol{N}$ for all 1000 choices for $\boldsymbol{y}$.
Reference values are obtained by a Monte Carlo $\text{MC}(\frac{\varepsilon}{10})$
simulation. We observe that the error is less than $\varepsilon$
in \emph{all} cases. The average CPU time to compute the CDF for a
single $\boldsymbol{y}$ is 0.019sec. We report some $\boldsymbol{y}$
corresponding to quantiles far away from the median and their COS-i
approximations and reference values in Table \ref{tab:Nopt}. This
experiment indicates that a CDF can be approximated on the entire
$\mathbb{R}^{d}$ by the COS-i method. 

\begin{table}[H]
\begin{centering}
\begin{tabular}{|c|c|c|}
\hline 
$\boldsymbol{y}$ & COS-i approximation & Ref. value\tabularnewline
\hline 
\hline 
{\small$(-0.49,0.18,0.3)$} & 0.0103 & {\small 0.0103}\tabularnewline
\hline 
{\small$(-0.02,-0.02,0.27)$} & 0.2505 & {\small 0.2505}\tabularnewline
\hline 
{\small$(0.07,0.21,0.15)$} & 0.5096 & {\small 0.5096}\tabularnewline
\hline 
$(0.30,0.26,0.17)$ & 0.7509 & {\small 0.7508}\tabularnewline
\hline 
$(0.94,0.89,0.45)$ & 0.9907 & {\small 0.9907}\tabularnewline
\hline 
\end{tabular}
\par\end{centering}
\caption{\protect\label{tab:Nopt}Approximation of a $\text{VG}(a,s,\boldsymbol{\eta},\boldsymbol{\theta},\boldsymbol{\sigma})$
distribution at different $\boldsymbol{y}\in\mathbb{R}^{3}$.}
\end{table}

Next, we approximate the CDF for a multivariate normal distribution
in $d=4$ dimensions with location zero and covariance matrix $\Sigma_{ii}=1$,
$\Sigma_{ij}=\rho$, $i\neq j$ for $\rho=0.75$. We set $\varepsilon=10^{-2}$.
We choose the truncation range by Inequality (\ref{eq:m-1-1}) and
the number of terms according to Corollary \ref{cor:NrTerms}. We
obtain $\boldsymbol{L}=(4.34,...,4.34)$ and $\boldsymbol{N}=(29,...,29)$.
We randomly choose 1000 different $\boldsymbol{y}\in\mathbb{R}^{d}$
according to the multivariate normal distribution with location zero
and covariance matrix $\Sigma$. Then, we approximate the multivariate
normal CDF at each $\boldsymbol{y}$ by the classical COS-i method.
The error tolerance is satisfied in \emph{all} cases. Reference values
are obtained by a Monte Carlo $\text{MC}(\frac{\varepsilon}{10})$
simulation. We repeat this experiment with other choices for $\rho$,
i.e., $\rho\in\{0.0,0.5,0.9,0.99\}$, and observe that $\boldsymbol{L}$
and $\boldsymbol{N}$ do not change and that the error tolerance is
satisfied in all cases.

Lastly, we test the damped COS-iv method in the BS model to approximate
the price of an unweighted arithmetic put option. We consider the
case $d=2$ and set $T=1$, $r=0$, $\boldsymbol{S}(0)=(50,50)$,
$K=100$ and covariance matrix $\Sigma\in\mathbb{R}^{d\times d}$
with $\Sigma_{11}=0.2^{2},$ $\Sigma_{22}=0.4^{2}$, $\Sigma_{12}=\Sigma_{21}=\frac{1}{2}\sqrt{\Sigma_{11}\Sigma_{22}}$.
We obtain the truncation range $\boldsymbol{L}=(3.9,7.9)$ from Inequality
(\ref{eq:m-1-1}) using $n=8$ moments and error tolerance $\varepsilon=10^{-2}$.
We set $\boldsymbol{\alpha}=(-4,-4)$. Applying Corollary \ref{cor:NrTerms},
we set $\boldsymbol{N}=(72,72)$. Reference values are obtained by
a Monte Carlo $\text{MC}(\frac{\varepsilon}{10})$ simulation. The
damped COS-iv method prices the unweighted arithmetic put option in
0.007sec. The minimal number of terms to obtain the required error
tolerance is roughly twice as large as $\boldsymbol{N}$ in this example
and is given by $(40,40)$, which leads to a CPU time of 0.002sec.
We average over ten runs to obtain the CPU time. 

\subsection{\protect\label{subsec:Comparison-with-Monte}Comparison with Monte
Carlo and Lewis}

We compare the classical COS-i and the damped COS-iv methods with
a Monte Carlo (MC) simulation in the VG and BS models to obtain the
price of a digital cash-or-nothing put option or the price of an unweighted
arithmetic basket put option. According to Example \ref{exa:(Various-other-European},
the digital cash-or-nothing put option can also be interpreted as
a CDF. 

The computational complexity of a MC simulation with $U\in\mathbb{N}$
runs scales like $O(U)$. The classical COS-i and damped COS-iv methods
consist of $d$-nested sums. According to Equation (\ref{eq:ck}),
the computational complexity of the COS-i and COS-iv methods scale
like $O\left(\prod_{h=1}^{d}\{N_{h}\}\right)$. A MC simulation converges
relatively slowly, but it scarcely depends on the dimension. On the
other hand, the complexity of the COS-i and COS-iv methods grows exponentially
in the dimension; however, the COS-i and COS-iv methods also converge
exponentially for the multivariate normal distribution. The choice
between MC and the COS method depends both on the dimension and on
the error tolerance $\varepsilon$: the higher $d$, the better MC
compares to the COS method, but the smaller $\varepsilon$, the faster
the COS method performs.

In Table \ref{tab:=000020COS=000020-=000020MC=000020-=000020BS} (Table
\ref{tab:=000020COS=000020-=000020MC=000020-=000020VG}), we price
these two options in $d\in\{2,4\}$ dimensions in the BS model (VG
model). We observe that the classical COS-i and the damped COS-iv
methods in $d=2$ dimensions are significantly faster than MC, approximately
by a factor between $500$ and $5000$. In $d=4$ dimensions, the
COS-i method outperforms MC but the damped COS-iv method is slower
than a MC simulation. The standard deviation of a single Monte Carlo
run decreases in the examples discussed in Tables \ref{tab:=000020COS=000020-=000020MC=000020-=000020BS}
and \ref{tab:=000020COS=000020-=000020MC=000020-=000020VG} with increasing
dimension, which explains why the CPU time of the Monte Carlo method
decreases as the dimension increases.

\begin{table}[H]
\begin{centering}
\begin{tabular}{|c|c|c|c|c|c|c|>{\centering}p{1.5cm}|>{\centering}p{1.7cm}|>{\centering}p{1.5cm}|}
\hline 
$d$ & Function of interest & $N$ & $L$ & $\alpha$ & $U$ & $\sigma_{MC}$ & Ref. value & CPU time COS & CPU time MC\tabularnewline
\hline 
\hline 
2 & Cash-or-nothing put & $5$ & $0.796$ & $0$ & $15392$ & 0.48 & $0.3741$ & $1.61\cdot10^{-5}$ & 0.008\tabularnewline
\hline 
4 & Cash-or-nothing put & $10$ & $0.868$ & $0$ & $11975$ & 0.42 & $0.2345$ & $0.0072$ & 0.01\tabularnewline
\hline 
2 & Arithm. basket put & $25$ & $2.585$ & $-3$ & $5070844$ & 9.16 & $6.9066$ & $0.0015$ & 1.49\tabularnewline
\hline 
4 & Arithm. basket put & $35$ & $4.688$ & $-1.5$ & $4607857$ & 8.38 & $6.3056$ & $13.873$ & 3.44\tabularnewline
\hline 
\end{tabular}
\par\end{centering}
\caption{\protect\label{tab:=000020COS=000020-=000020MC=000020-=000020BS}Let
$\varepsilon=10^{-2}$. Comparison of a Monte Carlo $\text{MC}(\varepsilon)$
simulation to the classical COS-i method to obtain the price of a
cash-or-nothing digital put option and the damped COS-iv method to
obtain the price of an unweighted arithmetic basket put option in
the BS model. We set $T=1$, $r=0$ and $\Sigma_{ii}=0.2^{2}$ and
$\Sigma_{ij}=\frac{1}{2}\cdot0.2^{2}$ for $j\protect\neq i$. For
the cash-or-nothing put option, we set $\boldsymbol{S}_{0}=\boldsymbol{K}=(100,\dots,100)$
and for the unweighted arithmetic basket put option, we set $K=100$
and $\boldsymbol{S}_{0}=\big(\frac{100}{d},...,\frac{100}{d}\big)$.
We obtain the truncation range $\boldsymbol{L}=(L,\dots,L)$ from
Inequality (\ref{eq:m-1-1}) using $n=8$ moments. The number of terms
$\boldsymbol{N}$ and for the basket option and the damping factor
$\boldsymbol{\alpha}=(\alpha,\dots,\alpha)$, are chosen, such that
$\boldsymbol{N}$ is minimal so that the error of the COS-i and COS-iv
methods are smaller than $\varepsilon$. The number of runs for the
Monte Carlo simulation is denoted by $U$. The standard deviation
of a single Monte Carlo run is denoted by $\sigma_{MC}$. Reference
values are calculated by a Monte Carlo $\text{MC}(\frac{\varepsilon}{10})$
simulation. We average over 10 runs to obtain the CPU time, which
is measured in seconds. }
\end{table}

\begin{table}[H]
\begin{centering}
\begin{tabular}{|c|c|c|c|c|c|c|>{\centering}p{1.5cm}|>{\centering}p{1.7cm}|>{\centering}p{1.5cm}|}
\hline 
$d$ & Function of interest & $N$ & $L$ & $\alpha$ & $U$ & $\sigma_{MC}$ & Ref. value & CPU time COS & CPU time MC\tabularnewline
\hline 
\hline 
2 & Cash-or-nothing put & $5$ & $0.853$ & $0$ & $13728$ & 0.45 & $0.2898$ & $2.79\cdot10^{-5}$ & 0.015\tabularnewline
\hline 
4 & Cash-or-nothing put & $5$ & $0.930$ & $0$ & $5261$ & 0.27 & $0.0839$ & $0.0027$ & 0.005\tabularnewline
\hline 
2 & Arithm. basket put & $20$ & $2.581$ & $-2.5$ & $4002494$ & 7.66 & $5.5951$ & $0.0014$ & 1.5\tabularnewline
\hline 
4 & Arithm. basket put & $30$ & $5.029$ & $-1.5$ & $2133811$ & 5.82 & $3.9696$ & $16.355$ & 1.3\tabularnewline
\hline 
\end{tabular}
\par\end{centering}
\caption{\protect\label{tab:=000020COS=000020-=000020MC=000020-=000020VG}Let
$\varepsilon=10^{-2}$. Comparison of a Monte Carlo $\text{MC}(\varepsilon)$
simulation to the classical COS-i method to obtain the price of a
cash-or-nothing digital put option and the damped COS-iv method to
obtain the price of an unweighted arithmetic basket put option in
the VG model. We set $\boldsymbol{\sigma}=(0.2,\dots,0.2)$, $\boldsymbol{\theta}=(-0.03,\dots,-0.03)$,
$\nu=0.1$ and choose the other parameters as in Table \ref{tab:=000020COS=000020-=000020MC=000020-=000020BS}.}
\end{table}

We also compare the CPU time of the damped COS-iv method for very
small error tolerance $\varepsilon$ to the CPU time of a Monte Carlo
simulation. For tiny $\varepsilon$, the computation of reference
values by a Monte Carlo simulation would take too long, which is why
we provide several numerical experiments with an uncorrelated multivariate
normal distribution, where reference value are given in closed form.
One could and should use the classical COS-i method to approximate
the multivariate normal CDF, but we explicitly want to test the damped
COS-iv method. We observe that the damped COS-iv method is faster
than MC for $d\leq4$ and $\varepsilon\leq10^{-3}$. For $d=5$, the
damped COS-iv method outperforms MC for $\varepsilon\leq10^{-5}$;
otherwise, a MC simulation is faster. If $\varepsilon=10^{-9}$ and
$d=4$, the damped COS-iv method needs $220$ terms in each dimension
to stay below the error tolerance and the CPU time is about one hour.
We estimate that a MC simulation would take longer than $20,000$
years. Some experiments are also shown in Figure \ref{fig:Computational-cost-of}.

\begin{figure}[H]
\begin{centering}
\includegraphics[scale=0.4]{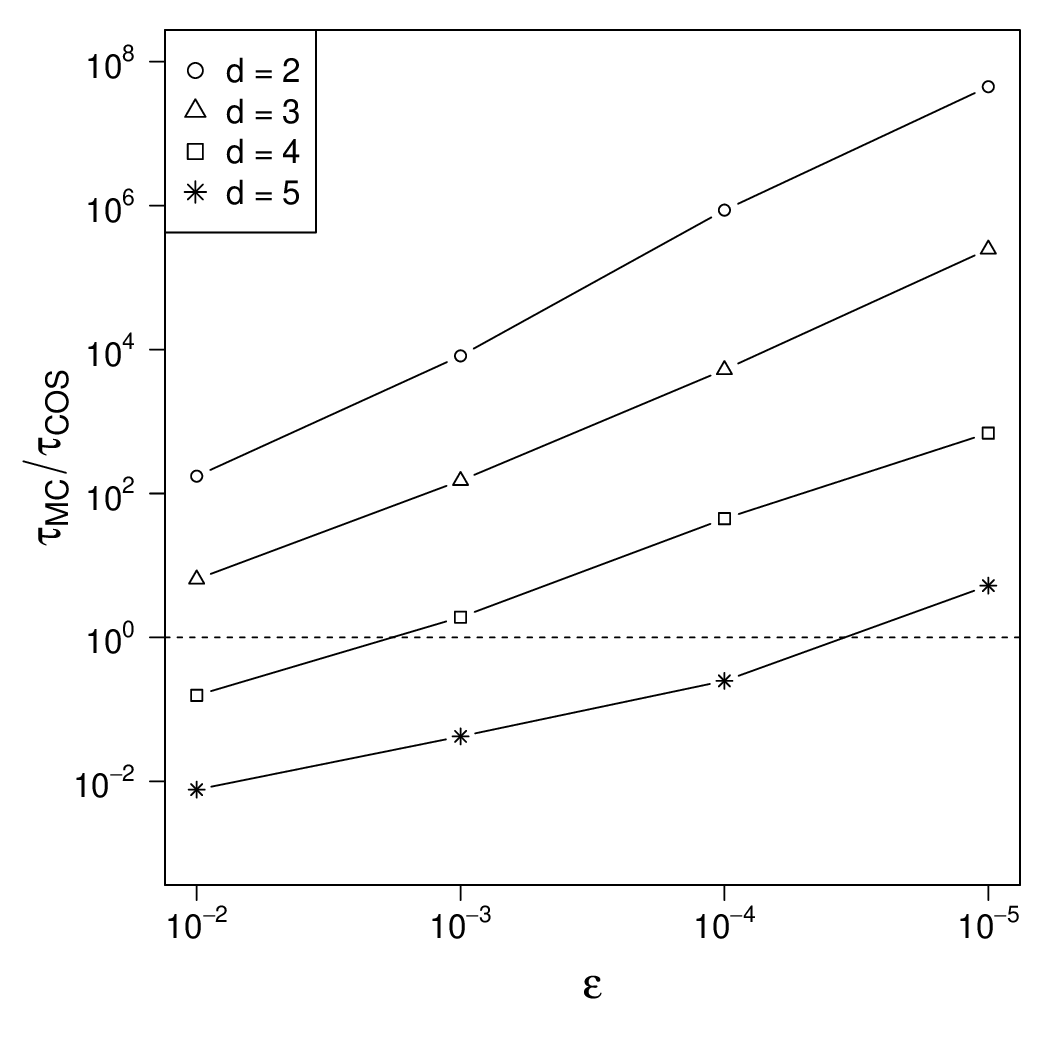}
\par\end{centering}
\caption{\protect\label{fig:Computational-cost-of}Ratio of the CPU time of
the damped COS-iv method $(\tau_{\text{COS}})$ and the CPU time of
a MC simulation $(\tau_{\text{MC}})$ to approximate the CDF of a
multivariate normal distribution. We set $\Sigma_{ii}=\sigma^{2}$,
$\Sigma_{ij}=0$, $i\protect\neq j$, where $\sigma=0.2$, and we
use $\boldsymbol{\eta}=(4.58517,...,4.58517)$, $\boldsymbol{y}=(4.60517,...,4.60517)$
and $\boldsymbol{\alpha}=(-7,\dots,-7)$. We obtain the truncation
range $\boldsymbol{L}=(L,...,L)$ from Inequality (\ref{eq:m-1-1})
using $n=8$ moments. We use the minimal number of terms $\boldsymbol{N}$
to stay below the error tolerance. The reference value can be obtained
in closed form. }
\end{figure}

The Lewis formula (\citet[Theorem 3.2]{eberlein2010analysis} and
\citet[Proposition 2.4]{bayer2022optimal}) expresses the integral
in (\ref{eq:int}) by another integral involving the Fourier-transforms
$\widehat{g}$ and $\widehat{w}$. The computational speed and robustness
of the Lewis formula strongly depend on the numerical method for approximating
the new integral. We implement the Lewis formula using Gauss-Hermite
quadrature and the R package \emph{mvQuad} without parallelization.
Using this package, the CPU time of the Lewis formula is of the same
order of magnitude as the damped COS-iv method in two dimensions and
is usually faster in four dimensions, see Table \ref{tab:Lewis}. 

\begin{table}[H]
\begin{centering}
\begin{tabular}{|>{\centering}p{0.1cm}|>{\centering}p{0.8cm}|>{\centering}p{0.55cm}|>{\centering}p{1.1cm}|>{\centering}p{0.75cm}|>{\centering}p{0.75cm}|>{\centering}p{0.75cm}|>{\centering}p{2.5cm}|c|>{\centering}p{2.3cm}|>{\centering}p{1.75cm}|}
\hline 
$d$ & Model & $\varepsilon$ & Ref. value & CPU MC & CPU Lewis & CPU COS & $\boldsymbol{L}$ & $N$ & $\boldsymbol{\alpha}$ & Damping Lewis\tabularnewline
\hline 
\hline 
2 & BS & $10^{-4}$ & 11.4474 & $1.7$ & $0.003$ & $0.001$ & $(6.09,6.09)$ & $24$ & $(-0.9,-0.9)$ & $(2.5,2.5)$\tabularnewline
\hline 
4 & BS & $10^{-4}$ & 8.193 & $39$ & $0.52$ & $48.57$ & $(9.35,\dots,9.35)$ & $42$ & $(-0.7,\dots,-0.7)$ & $(2.1,...,2.1)$\tabularnewline
\hline 
4 & BS & $10^{-4}$ & 3.1 & $13.6$ & $26.7$ & $12.03$ & $(4.57,9.15,$

$13.72,18.29)$ & $33$ & $(-0.9,-0.7,$

$-0.6,-0.4)$ & $(2.4,1.9,$

$1.5,1.2)$\tabularnewline
\hline 
2 & VG & $10^{-3}$ & 11.7589 & $16.4$ & $0.001$ & $0.004$ & $(8.14,8.14)$ & $26$ & $(-1,-1)$ & $(1.7,1.7)$\tabularnewline
\hline 
4 & VG & $10^{-4}$ & 8.9441 & $40.8$ & $0.69$ & $70.34$ & $(12.80,\dots,12.80)$ & $44$ & $(-0.5,\dots,-0.5)$ & $(1.2,...,1.2)$\tabularnewline
\hline 
\end{tabular}
\par\end{centering}
\caption{\protect\label{tab:Lewis}Comparison of the damped COS-iv method and
the Lewis formula to price an unweighted arithmetic basket put option
with $K=100$, $\boldsymbol{S}_{0}=\big(\frac{100}{d},...,\frac{100}{d}\big)$,
$T=1$ and $r=0$. Reference values, optimal damping factors for the
Lewis formula and number of runs for the Monte Carlo simulation are
taken from \citet[Tab. 4.6]{bayer2022optimal}, who used the following
parameters for the BS model: $\Sigma_{ii}=0.4^{2}$ (first two rows),
$\Sigma_{ii}=\big(\frac{2i}{10}\big)^{2}$ (third row) and $\Sigma_{ij}=0$,
$i\protect\neq j$. For the VG model, they set $\boldsymbol{\sigma}=(0.4,...,0.4)$,
$\boldsymbol{\theta}=(-0.3,...,-0.3)$ and $\nu=0.257$. We obtain
the truncation range $\boldsymbol{L}$ for the COS-iv method from
Inequality (\ref{eq:m-1-1}) using $n=8$ moments and setting the
absolute error tolerance $\varepsilon$ (approximately) equal to the
statistical error of the reference values taken from \citet[Table 4.1]{bayer2022optimal}.
The number of terms $\boldsymbol{N}=(N,\dots,N)$ and the damping
factor $\boldsymbol{\alpha}$ are chosen so that $\boldsymbol{N}$
is minimal and the relative errors of the COS-iv method are smaller
or equal to the relative errors of the Lewis formula, which are of
the order $10^{-4}$. We average over 10 runs to obtain the CPU time,
which is measured in seconds.}
\end{table}

\section{\protect\label{sec:Conclusions}Conclusion}

In this article we introduced and discussed in general dimensions
the classical and the damped COS methods, which are tools for solving
certain multidimensional integrals numerically, e.g., approximating
the CDF from a characteristic function or pricing a financial contract.
We showed that the classical and the damped COS methods converge exponentially
if the Fourier transform of the (damped) density decays exponentially. 

We compared the classical and the damped COS methods to other numerical
techniques, such as the Lewis formula and Monte Carlo simulations.
In terms for speed, we observed that the classical and damped COS
methods are faster than a crude Monte Carlo simulation in two and
three dimensions but slower in five dimensions for an absolute error
tolerance of about $10^{-2}$. The choice between the COS method and
a Monte Carlo simulation in four dimensions depends on the required
error tolerance and the function of interest.

The Lewis formula requires solving an integral involving the Fourier
transforms of the density and the function of interest numerically.
The computational speed of the Lewis formula depends strongly on the
numerical method used to solve that integral. In our experiments,
the CPU time of the Lewis formula is of the same order of magnitude
as the damped COS-iv method in two dimensions and usually faster in
four dimensions.

Let us compare the three methods in terms of error control: The number
of runs for a crude Monte Carlo simulation is given explicitly thanks
to the central limit theorem. The classical COS method requires two
parameters: a truncation range and a number of terms. We developed
explicit and implicit formulas for both parameters. A damping parameter
needs to be provided to apply the damped COS method. So far, we are
only able to find the damping factor by trial and error. The Lewis
formula also requires a damping factor, which can be obtained via
an optimization techniques, see \citet{bayer2022optimal}. However,
the Lewis formula requires a numerical integration method and hence
some tuning parameters such as a number of quadrature points, etc.,
which are not always easy to estimate.

\section*{Acknowledgments}

We thank two anonymous referees for very helpful comments that improved
this paper. 

\section*{Declarations}

The authors declare no competing interests.

\section*{Availability of data, materials and Code}

No data was used for the research described in the article. The code
for the experiments presented in this paper is available online. URLs
are included in this article.

\appendix

\section{Auxiliary results}
\begin{prop}
\label{prop:COSadmissible}Assume $\psi\in\mathcal{L}^{1}\cap\mathcal{L}^{2}$.
Then $\psi$ is COS-admissible if Inequality (\ref{eq:propBed}) holds.
Let $\boldsymbol{L}=(L_{1},...,L_{d})\in\mathbb{R}_{+}^{d}$; then
it holds that
\begin{align}
B_{\psi}(\boldsymbol{L}) & \leq\Xi\int_{\mathbb{R}^{d}\setminus[-\boldsymbol{L},\boldsymbol{L}]}\prod_{h=1}^{d}\max\left\{ x_{h}^{2}L_{h}^{-2},1\right\} |\psi(\boldsymbol{x})|^{2}d\boldsymbol{x}\label{eq:B(L)_1}\\
 & \leq\frac{\Xi}{d\underset{h=1,...,d}{\min}L_{h}^{2d}}\int_{\mathbb{R}^{d}\setminus[-\boldsymbol{L},\boldsymbol{L}]}\left|\boldsymbol{x}\right|^{2d}|\psi(\boldsymbol{x})|^{2}d\boldsymbol{x}+\Xi\int_{\mathbb{R}^{d}\setminus[-\boldsymbol{L},\boldsymbol{L}]}|\psi(\boldsymbol{x})|^{2}d\boldsymbol{x},\label{eq:B(L)_2}
\end{align}
where $\Xi=\frac{\pi^{2}}{3}\sum_{h=1}^{d}\left(\frac{\pi^{2}}{3}+1\right)^{h-1}$.
\end{prop}

\begin{proof}
Let $\boldsymbol{L}\in\mathbb{R}_{+}^{d}$ and $\boldsymbol{j}\in\mathbb{Z}^{d}$.
It follows by Parseval's identity
\begin{align}
\int_{[2\boldsymbol{jL}-\boldsymbol{L},2\boldsymbol{jL}+\boldsymbol{L}]}|\psi(\boldsymbol{x})|^{2}d\boldsymbol{x}= & \sideset{}{'}\sum_{\boldsymbol{k}\in\mathbb{N}_{0}^{d}}\frac{1}{\prod_{h=1}^{d}L_{h}}\bigg|\int_{[2\boldsymbol{jL}-\boldsymbol{L},2\boldsymbol{jL}+\boldsymbol{L}]}\psi(\boldsymbol{x})\prod_{h=1}^{d}\underbrace{\cos\left(k_{h}\pi\frac{x_{h}-(2j_{h}L_{h}-L_{h})}{2L_{h}}\right)}_{=(-1)^{j_{h}k_{h}}\cos\left(k_{h}\pi\frac{x_{h}+L_{h}}{2L_{h}}\right)}d\boldsymbol{x}\bigg|^{2}\nonumber \\
= & \sideset{}{'}\sum_{\boldsymbol{k}\in\mathbb{N}_{0}^{d}}\frac{1}{\prod_{h=1}^{d}L_{h}}\bigg|\int_{[2\boldsymbol{jL}-\boldsymbol{L},2\boldsymbol{jL}+\boldsymbol{L}]}\psi(\boldsymbol{x})e_{\boldsymbol{k}}(\boldsymbol{x})d\boldsymbol{x}\bigg|^{2}.\label{eq:parseval}
\end{align}
By the Cauchy-Schwarz inequality, we obtain with $\varrho(\boldsymbol{j}):=\prod_{h=1}^{d}\max\{|j_{h}|,1\}$,
\begin{align}
\left|\int_{\mathbb{R}^{d}\setminus[-\boldsymbol{L},\boldsymbol{L}]}\psi(\boldsymbol{x})e_{\boldsymbol{k}}(\boldsymbol{x})d\boldsymbol{x}\right|^{2} & =\left|\sum_{\boldsymbol{j}\in\mathbb{Z}^{d}\setminus\{\boldsymbol{0}\}}\frac{\varrho(\boldsymbol{j})}{\varrho(\boldsymbol{j})}\int_{[2\boldsymbol{jL}-\boldsymbol{L},2\boldsymbol{jL}+\boldsymbol{L}]}\psi(\boldsymbol{x})e_{\boldsymbol{k}}(\boldsymbol{x})d\boldsymbol{x}\right|^{2}\nonumber \\
 & \leq\bigg(\underbrace{\sum_{\boldsymbol{j}\in\mathbb{Z}^{d}\setminus\{\boldsymbol{0}\}}\frac{1}{(\varrho(\boldsymbol{j}))^{2}}}_{=\Xi}\bigg)\sum_{\boldsymbol{j}\in\mathbb{Z}^{d}\setminus\{\boldsymbol{0}\}}(\varrho(\boldsymbol{j}))^{2}\left|\int_{[2\boldsymbol{jL}-\boldsymbol{L},2\boldsymbol{jL}+\boldsymbol{L}]}\psi(\boldsymbol{x})e_{\boldsymbol{k}}(\boldsymbol{x})d\boldsymbol{x}\right|^{2}.\label{eq:CS}
\end{align}
The fact that $\Xi=\frac{\pi^{2}}{3}\sum_{h=1}^{d}\left(\frac{\pi^{2}}{3}+1\right)^{h-1}$
can be shown by mathematical induction over $d$. Then it follows
that
\begin{align*}
B_{\psi}(\boldsymbol{L}) & \overset{(\ref{eq:CS})}{\leq}\,\Xi\sum_{\boldsymbol{j}\in\mathbb{Z}^{d}\setminus\{\boldsymbol{0}\}}(\varrho(\boldsymbol{j}))^{2}\sideset{}{'}\sum_{\boldsymbol{k}\in\mathbb{N}_{0}^{d}}\frac{1}{\prod_{h=1}^{d}L_{h}}\left|\int_{[2\boldsymbol{jL}-\boldsymbol{L},2\boldsymbol{jL}+\boldsymbol{L}]}\psi(\boldsymbol{x})e_{\boldsymbol{k}}(\boldsymbol{x})d\boldsymbol{x}\right|^{2}\\
 & \overset{(\ref{eq:parseval})}{=}\Xi\sum_{\boldsymbol{j}\in\mathbb{Z}^{d}\setminus\{\boldsymbol{0}\}}(\varrho(\boldsymbol{j}))^{2}\int_{[2\boldsymbol{jL}-\boldsymbol{L},2\boldsymbol{jL}+\boldsymbol{L}]}|\psi(\boldsymbol{x})|^{2}d\boldsymbol{x}.
\end{align*}
For $\boldsymbol{j}\in\mathbb{Z}^{d}$ and $\boldsymbol{x}\in[2\boldsymbol{jL}-\boldsymbol{L},2\boldsymbol{jL}+\boldsymbol{L}]$,
one has $|j_{h}|\leq\frac{|x_{h}|}{L_{h}}$, $h=1,...,d$. It follows
that $(\varrho(\boldsymbol{j)})^{2}\leq\prod_{h=1}^{d}\max\left\{ x_{h}^{2}L_{h}^{-2},1\right\} $,
which implies Inequality (\ref{eq:B(L)_1}). By Young's inequality,
it holds that 
\begin{align*}
\prod_{h=1}^{d}\left(\max\left\{ x_{h}^{2d}L_{h}^{-2d},1\right\} \right)^{\frac{1}{d}} & \leq\frac{1}{d}\sum_{h=1}^{d}\max\left\{ x_{h}^{2d}L_{h}^{-2d},1\right\} \leq\frac{\left|\boldsymbol{x}\right|^{2d}}{d\underset{h=1,...,d}{\min}L_{h}^{2d}}+1.
\end{align*}
In the last inequality, we used $\max\{a,b\}\leq a+b$ for any $a,b\geq0$
and $\sum_{h=1}^{d}x_{h}^{2d}\leq\left|\boldsymbol{x}\right|^{2d}$,
which follows from the monotonicity of the $p$-norm. Hence, Inequality
(\ref{eq:B(L)_2}) holds. Assumption (\ref{eq:propBed}) and $\psi\in\mathcal{L}^{2}$
imply $B_{\psi}(\boldsymbol{L})\to0$, $\min_{h=1,...,d}L_{h}\to\infty$.
\end{proof}
\begin{prop}
\label{prop:centr}Let $\boldsymbol{\alpha}\in\mathbb{R}^{d}$. Assume
that $g$ is a density and that $\boldsymbol{x}\mapsto|\boldsymbol{x}|e^{\boldsymbol{\alpha}\cdot\boldsymbol{x}}g(\boldsymbol{x})$
is integrable. Let $\lambda=(\widehat{g}(-i\boldsymbol{\alpha}))^{-1}$
then $\lambda\in(0,\infty)$. Choose $\boldsymbol{\mu}\in\mathbb{R}^{d}$
by
\begin{equation}
\mu_{h}=-\lambda i\left.\frac{\partial}{\partial u_{h}}\widehat{g}(\boldsymbol{u}-i\boldsymbol{\alpha})\right|_{\boldsymbol{u}=\boldsymbol{0}},\quad h=1,...,d.\label{eq:mu_lam_partial}
\end{equation}
Define $f$ as in Equation (\ref{eq:dampedf}). Then $f$ is a density
with characteristic function 
\begin{equation}
\widehat{f}(\boldsymbol{u})=\lambda e^{-i\boldsymbol{u}\cdot\boldsymbol{\mu}}\widehat{g}(\boldsymbol{u}-i\boldsymbol{\alpha}),\quad\boldsymbol{u}\in\mathbb{R}^{d}.\label{eq:Fourier_f}
\end{equation}
Further, the moments of $f$ of first order are zero, i.e., $\int_{\mathbb{R}^{d}}f(\boldsymbol{x})x_{h}d\boldsymbol{x}=0$,
$h=1,...,d$.
\end{prop}

\begin{proof}
Use $\int|\boldsymbol{x}|e^{\boldsymbol{\alpha}\cdot\boldsymbol{x}}g(\boldsymbol{x})d\boldsymbol{x}<\infty$
and split the integration range into $\mathbb{R}^{d}\setminus B_{1}$
and $B_{1}$, where $B_{1}$ is the unit ball, to see that $\boldsymbol{x}\mapsto e^{\boldsymbol{\alpha}\cdot\boldsymbol{x}}g(\boldsymbol{x})$
is integrable. Since $\lambda=(\int_{\mathbb{R}^{d}}e^{\boldsymbol{\alpha}\cdot\boldsymbol{x}}g(\boldsymbol{x})d\boldsymbol{x})^{-1}$
and $g$ is a density we have $\lambda\in(0,\infty)$. By the definition
of $\lambda$, $f$ is a density. Since $f\in\mathcal{L}^{1}$, $\widehat{f}$
exists. Note that 
\[
\widehat{f}(\boldsymbol{u})=\int_{\mathbb{R}^{d}}\lambda e^{\boldsymbol{\alpha}\cdot(\boldsymbol{x}+\boldsymbol{\mu})}g(\boldsymbol{x}+\boldsymbol{\mu})e^{i\boldsymbol{u}\cdot\boldsymbol{x}}d\boldsymbol{x}=\lambda\int_{\mathbb{R}^{d}}e^{\boldsymbol{\alpha}\cdot\boldsymbol{x}}g(\boldsymbol{x})e^{i\boldsymbol{u}\cdot(\boldsymbol{x}-\boldsymbol{\mu})}d\boldsymbol{x}=\lambda e^{-i\boldsymbol{\mu}\cdot\boldsymbol{u}}\int_{\mathbb{R}^{d}}g(\boldsymbol{x})e^{i(\boldsymbol{u}-i\boldsymbol{\alpha})\cdot\boldsymbol{x}}d\boldsymbol{x},
\]
which implies Equation (\ref{eq:Fourier_f}). By \citet[Thm 25.2]{bauer1996probability},
the partial derivatives in Equation (\ref{eq:mu_lam_partial}) exist
and it holds that $\mu_{h}=\lambda\int_{\mathbb{R}^{d}}e^{\boldsymbol{\alpha}\cdot\boldsymbol{x}}g(\boldsymbol{x})x_{h}d\boldsymbol{x}$.
Finally, we have that
\begin{align*}
\int_{\mathbb{R}^{d}}f(\boldsymbol{x})x_{h}d\boldsymbol{x}= & \int_{\mathbb{R}^{d}}\lambda e^{\boldsymbol{\alpha}\cdot(\boldsymbol{x}+\boldsymbol{\mu})}g(\boldsymbol{x}+\boldsymbol{\mu})x_{h}d\boldsymbol{x}=\lambda\int_{\mathbb{R}^{d}}e^{\boldsymbol{\alpha}\cdot\boldsymbol{x}}g(\boldsymbol{x})x_{h}d\boldsymbol{x}-\mu_{h}\lambda\int_{\mathbb{R}^{d}}e^{\boldsymbol{\alpha}\cdot\boldsymbol{x}}g(\boldsymbol{x})d\boldsymbol{x}=0.
\end{align*}
\end{proof}
\begin{lem}
\label{lem:fMinusAkEk}Let $\psi\in\mathcal{L}^{1}\cap\mathcal{L}^{2}$.
Let $\boldsymbol{M},\boldsymbol{L}\in\mathbb{R}_{+}^{d}$ with $\boldsymbol{M}\leq\boldsymbol{L}$;
then it holds that
\[
\Vert\psi1_{[-\boldsymbol{L},\boldsymbol{L}]}-\sideset{}{'}\sum_{\boldsymbol{\boldsymbol{0}}\leq\boldsymbol{k}\leq\boldsymbol{N}}a_{\boldsymbol{k}}e_{\boldsymbol{k}}1_{[-\boldsymbol{L},\boldsymbol{L}]}\Vert_{2}^{2}\leq\int_{\mathbb{R}^{d}}|\psi(\boldsymbol{x})|^{2}d\boldsymbol{x}-\prod_{h=1}^{d}L_{h}\sideset{}{'}\sum_{\boldsymbol{\boldsymbol{0}}\leq\boldsymbol{k}\leq\boldsymbol{N}}|c_{\boldsymbol{k}}|^{2}+G(\boldsymbol{L}),
\]
where 
\begin{equation}
G(\boldsymbol{L}):=B_{\psi}(\boldsymbol{L})+2\sqrt{B_{\psi}(\boldsymbol{L})\int_{\mathbb{R}^{d}}|\psi(\boldsymbol{x})|^{2}d\boldsymbol{x}}.\label{eq:eta_l}
\end{equation}
\end{lem}

\begin{proof}
Let
\[
\phi_{\boldsymbol{k}}:=\frac{1}{\prod_{h=1}^{d}L_{h}}\int_{\mathbb{R}^{d}\setminus[-\boldsymbol{L},\boldsymbol{L}]}\psi(\boldsymbol{x})e_{\boldsymbol{k}}(\boldsymbol{x})d\boldsymbol{x},\quad\boldsymbol{k}\in\mathbb{N}_{0}^{d}.
\]
It holds that $c_{\boldsymbol{k}}=a_{\boldsymbol{k}}+\phi_{\boldsymbol{k}}$.
It follows by the Cauchy-Schwarz inequality that
\begin{align}
\prod_{h=1}^{d}L_{h}\sideset{}{'}\sum_{\boldsymbol{\boldsymbol{0}}\leq\boldsymbol{k}\leq\boldsymbol{N}}|c_{\boldsymbol{k}}|^{2}= & \prod_{h=1}^{d}L_{h}\left(\sideset{}{'}\sum_{\boldsymbol{\boldsymbol{0}}\leq\boldsymbol{k}\leq\boldsymbol{N}}|a_{\boldsymbol{k}}|^{2}+\sideset{}{'}\sum_{\boldsymbol{\boldsymbol{0}}\leq\boldsymbol{k}\leq\boldsymbol{N}}|\phi_{\boldsymbol{k}}|^{2}+2\sideset{}{'}\sum_{\boldsymbol{\boldsymbol{0}}\leq\boldsymbol{k}\leq\boldsymbol{N}}|\phi_{\boldsymbol{k}}||a_{\boldsymbol{k}}|\right)\nonumber \\
\leq & \prod_{h=1}^{d}L_{h}\sideset{}{'}\sum_{\boldsymbol{\boldsymbol{0}}\leq\boldsymbol{k}\leq\boldsymbol{N}}|a_{\boldsymbol{k}}|^{2}+B_{\psi}(\boldsymbol{L})+2\prod_{h=1}^{d}L_{h}\sqrt{\sideset{}{'}\sum_{\boldsymbol{\boldsymbol{0}}\leq\boldsymbol{k}\leq\boldsymbol{N}}|\phi_{\boldsymbol{k}}|^{2}\sideset{}{'}\sum_{\boldsymbol{\boldsymbol{0}}\leq\boldsymbol{k}\leq\boldsymbol{N}}|a_{\boldsymbol{k}}|^{2}}\nonumber \\
\overset{(\ref{eq:parseval})}{\leq} & \prod_{h=1}^{d}L_{h}\sideset{}{'}\sum_{\boldsymbol{\boldsymbol{0}}\leq\boldsymbol{k}\leq\boldsymbol{N}}|a_{\boldsymbol{k}}|^{2}+\underbrace{B_{\psi}(\boldsymbol{L})+2\sqrt{B_{\psi}(\boldsymbol{L})\int_{\mathbb{R}^{d}}|\psi(\boldsymbol{x})|^{2}d\boldsymbol{x}}}_{=G(\boldsymbol{L})}\label{eq:a_k_c_k}\\
\text{\ensuremath{\overset{(\ref{eq:parseval})}{\leq}}} & \int_{\mathbb{R}^{d}}|\psi(\boldsymbol{x})|^{2}d\boldsymbol{x}+G(\boldsymbol{L}).\label{eq:bound_ck2}
\end{align}
Hence,
\begin{align*}
\Vert\psi1_{[-\boldsymbol{L},\boldsymbol{L}]}-\sideset{}{'}\sum_{\boldsymbol{\boldsymbol{0}}\leq\boldsymbol{k}\leq\boldsymbol{N}}a_{\boldsymbol{k}}e_{\boldsymbol{k}}1_{[-\boldsymbol{L},\boldsymbol{L}]}\Vert_{2}^{2}\overset{(\ref{eq:ekel})}{\leq} & \prod_{h=1}^{d}L_{h}\sideset{}{'}\sum_{k_{1}>N_{1}\text{ or}\dots\text{or }k_{d}>N_{d}}|a_{\boldsymbol{k}}|^{2}\\
=\,\,\, & \prod_{h=1}^{d}L_{h}\sideset{}{'}\sum_{\boldsymbol{k}\in\mathbb{N}_{0}^{d}}|a_{\boldsymbol{k}}|^{2}-\prod_{h=1}^{d}L_{h}\sideset{}{'}\sum_{\boldsymbol{\boldsymbol{0}}\leq\boldsymbol{k}\leq\boldsymbol{N}}|a_{\boldsymbol{k}}|^{2}\\
\overset{(\ref{eq:parseval},\ref{eq:a_k_c_k})}{\leq} & \int_{\mathbb{R}^{d}}|\psi(\boldsymbol{x})|^{2}d\boldsymbol{x}-\prod_{h=1}^{d}L_{h}\sideset{}{'}\sum_{\boldsymbol{\boldsymbol{0}}\leq\boldsymbol{k}\leq\boldsymbol{N}}|c_{\boldsymbol{k}}|^{2}+G(\boldsymbol{L}).
\end{align*}
\end{proof}

\appendix

\section{R code and pseudocode}

\begin{algorithm}[H]
{\footnotesize d = 2}{\footnotesize\par}

{\footnotesize y = c(1.5, 1.5)}{\footnotesize\par}

{\footnotesize Sigma = matrix(c(1.0, 0.7, 0.7, 4.0), nrow = d)}{\footnotesize\par}

{\footnotesize eta = c(-1.0, 0)}{\footnotesize\par}

{\footnotesize N = c(40, 40)~~~~~~\#N can be estimated using
Alg. \ref{alg:Application-of-Lemma} }{\footnotesize\par}

{\footnotesize eps = 10\textasciicircum -3}{\footnotesize\par}

{\footnotesize lambda = function(alpha)\{exp(-sum(eta {*} alpha) -
0.5 {*} sum(alpha {*} Sigma \%{*}\% alpha) )\}~~~~~~\#Exa. \ref{exa:(Normal-distribution).-Let}}{\footnotesize\par}

{\footnotesize mu\_8 = 105 {*} diag(Sigma)\textasciicircum 4~~~~~~\#8th
marginal moments, Eq. (\ref{eq:moments}) }{\footnotesize\par}

{\footnotesize v\_inf\_norm = function(alpha)\{1 / lambda(alpha) {*}
exp(-sum(alpha {*} y))\}~~~~~~\#See Exa. \ref{exa:(Digital-cash-or-nothing-put}}{\footnotesize\par}

{\footnotesize M = function(alpha) (3 {*} d / eps {*} v\_inf\_norm(alpha)
{*} mu\_8)\textasciicircum (1 / 8)~~~~~~\#Eq. (\ref{eq:m-1-1})}{\footnotesize\par}

{\footnotesize L = function(alpha)\{M(alpha)\}}{\footnotesize\par}

{\footnotesize mu = function(alpha)\{eta + Sigma \%{*}\% alpha\}~~~~~~\#Exa.
\ref{exa:(Normal-distribution).-Let}}{\footnotesize\par}

{\footnotesize fhat = function(u)\{exp(-0.5 {*} sum(u {*} Sigma \%{*}\%
u))\}~~~~~~\#Exa. \ref{exa:(Normal-distribution).-Let}}{\footnotesize\par}

{\footnotesize tmp = function(s, k, psi)\{Re(psi(pi / 2 {*} s {*} k
/ L(alpha)) {*} exp(1i {*} pi / 2 {*} sum(s {*} k)))\}~~~~~~\#addend
for Eq. (\ref{eq:ck}, \ref{eq:vkTilde})}{\footnotesize\par}

{\footnotesize ~}{\footnotesize\par}

{\footnotesize\#\#\#\#\#\#\#\#\#\#\#\#\#\#\#\#\#\#\#\#\#\#\#\#\#\#\#\#
CLASSICAL COS-i METHOD\#\#\#\#\#\#\#\#\#\#\#\#\#\#\#\#\#\#\#\#\#}{\footnotesize\par}

{\footnotesize classical\_COS\_CDF = 0}{\footnotesize\par}

{\footnotesize alpha = rep(0,d)~~~~~~\#Damping factor}{\footnotesize\par}

{\footnotesize for(k1 in 0:N{[}1{]})\{}{\footnotesize\par}

{\footnotesize ~~~for(k2 in 0:N{[}2{]})\{}{\footnotesize\par}

{\footnotesize ~~~~~~k = c(k1, k2)}{\footnotesize\par}

{\footnotesize ~~~~~~ck = 1 / (2\textasciicircum (d - 1) {*}
prod(L(alpha))) {*} (tmp(c(1, 1), k, fhat) + tmp(c(1, -1), k, fhat))~~~~~~\#Eq.
(\ref{eq:ck})}{\footnotesize\par}

{\footnotesize ~~~~~~A = pmin(y - mu(rep(0, d)), M(alpha))~~~~~~\#Exa.
\ref{exa:CDF}}{\footnotesize\par}

{\footnotesize ~~~~~~tmp\_variable = (sin(k{[}k != 0{]} {*}
pi {*} (A{[}k != 0{]} + L(alpha){[}k != 0{]})/(2 {*} L(alpha){[}k
!= 0{]})))~~~~~~\#Using that M = L}{\footnotesize\par}

{\footnotesize ~~~~~~vk = prod(A{[}k == 0{]} + M(alpha){[}k
== 0{]}) {*} prod(2 {*} L(alpha){[}k != 0{]} / (k{[}k != 0{]} {*}
pi) {*} tmp\_variable) \#Eq. (\ref{eq:vkCDF})}{\footnotesize\par}

{\footnotesize ~~~~~~if(any((y - mu(rep(0, d))) < -M(alpha)))}{\footnotesize\par}

{\footnotesize ~~~~~~~~~vk = 0}{\footnotesize\par}

{\footnotesize ~~~~~~classical\_COS\_CDF = classical\_COS\_CDF
+ 1 / (2\textasciicircum sum(k == 0)) {*} ck {*} vk~~~~~~\#Eq.
(\ref{eq:vk})}{\footnotesize\par}

{\footnotesize ~~~~~~\}}{\footnotesize\par}

{\footnotesize\}}{\footnotesize\par}

{\footnotesize classical\_COS\_CDF~~~~~~\#0.7708859}{\footnotesize\par}

{\footnotesize ~}{\footnotesize\par}

{\footnotesize\#\#\#\#\#\#\#\#\#\#\#\#\#\#\#\#\#\#\#\#\#\#\#\#\#\#\#\#
DAMPED COS-iv METHOD\#\#\#\#\#\#\#\#\#\#\#\#\#\#\#\#\#\#\#\#\#}{\footnotesize\par}

{\footnotesize what = function(z)\{prod(exp(1i {*} y {*} z) / (1i {*}
z))\}~~~~~~\#Exa. \ref{exa:(Digital-cash-or-nothing-put}}{\footnotesize\par}

{\footnotesize vhat = function(u)\{1 / lambda(alpha) {*} exp(-1i {*}
sum(u {*} mu(alpha))) {*} what(u + 1i {*} alpha)\}~~~~~~\#Rem.
\ref{rem:fHat_vhat}}{\footnotesize\par}

{\footnotesize damped\_COS\_CDF = 0}{\footnotesize\par}

{\footnotesize alpha = rep(-1.0, d)~~~~~~\#Damping factor}{\footnotesize\par}

{\footnotesize for(k1 in 0:N{[}1{]})\{}{\footnotesize\par}

{\footnotesize ~~~for(k2 in 0:N{[}2{]})\{}{\footnotesize\par}

{\footnotesize ~~~~~~k = c(k1, k2)}{\footnotesize\par}

{\footnotesize ~~~~~~ck = 1 / (2\textasciicircum (d - 1) {*}
prod(L(alpha))) {*} (tmp(c(1, 1), k, fhat) + tmp(c(1, -1), k, fhat))~~~~~~\#Eq.
(\ref{eq:ck})}{\footnotesize\par}

{\footnotesize ~~~~~~vk\_tilde = 1 / (2\textasciicircum (d
- 1)) {*} (tmp(c(1, 1), k, vhat) + tmp(c(1, -1), k, vhat))~~~~~~\#Eq.
(\ref{eq:vkTilde})}{\footnotesize\par}

{\footnotesize ~~~~~~damped\_COS\_CDF = damped\_COS\_CDF + 1
/ (2\textasciicircum sum(k == 0)) {*} ck {*} vk\_tilde~~~~~~\#Eq.
(\ref{eq:vk_tilde})}{\footnotesize\par}

{\footnotesize ~~~~~~\}}{\footnotesize\par}

{\footnotesize\}}{\footnotesize\par}

{\footnotesize damped\_COS\_CDF~~~~~~\#0.7708836}{\footnotesize\par}

\caption{\protect\label{alg:Pure-R-code}Pure R code for the classical COS-i
and the damped COS-iv methods to obtain the CDF at $\boldsymbol{y}=(1.5,1.5)$
of a two dimensional normal distribution with location $\boldsymbol{\eta}=(-1,0)$
and covariance matrix $\Sigma$ with $\Sigma_{11}=1$, $\Sigma_{22}=4$,
$\Sigma_{12}=\Sigma_{21}=0.7$.}
\end{algorithm}

\begin{algorithm}[H]
{\footnotesize\texttt{Input: Error tolerance $\varepsilon>0$ and
the real number $I:=(2\pi)^{-d}\int_{\mathbb{R}^{d}}|\widehat{f}(\boldsymbol{u})|^{2}d\boldsymbol{u}$.}}~\\
{\footnotesize\texttt{Output: Approximation of $\int_{\mathbb{R}^{d}}v(\boldsymbol{x})f(\boldsymbol{x})d\boldsymbol{x}$
within the error tolerance $\varepsilon$.}}{\footnotesize\par}
\begin{lyxlist}{00.00.0000}
\item [{{\footnotesize\texttt{~}}}] {\footnotesize\texttt{DEFINE $\boldsymbol{L}$,
$\boldsymbol{M}$ via Eq. (\ref{eq:m-1-1})}}{\footnotesize\par}
\item [{{\footnotesize\texttt{~}}}] {\footnotesize\texttt{DEFINE $c_{\boldsymbol{k}}$
via Eq. (\ref{eq:ck})}}{\footnotesize\par}
\item [{{\footnotesize\texttt{~}}}] {\footnotesize\texttt{DEFINE $v_{\boldsymbol{k}}$
via Eq. (\ref{eq:vk})}}{\footnotesize\par}
\item [{{\footnotesize\texttt{~}}}] {\footnotesize\texttt{SET $\boldsymbol{k}:=\boldsymbol{0}\in\mathbb{N}_{0}^{d}$,
$\boldsymbol{N}:=\boldsymbol{0}\in\mathbb{N}_{0}^{d}$, $\gamma:=2^{-\Lambda(\boldsymbol{k})}|c_{\boldsymbol{k}}|^{2}$
and $\beta:=2^{-\Lambda(\boldsymbol{k})}c_{\boldsymbol{k}}v_{\boldsymbol{k}}$}}{\footnotesize\par}
\item [{{\footnotesize\texttt{~}}}] {\footnotesize\texttt{WHILE $\left|I-\gamma\prod_{h=1}^{d}L_{h}\right|>\frac{\varepsilon^{2}}{162\left\Vert v\right\Vert _{2}^{2}}$
DO:}}{\footnotesize\par}
\item [{{\footnotesize\texttt{~}}}] {\footnotesize\texttt{~~~FOR $\boldsymbol{k}\in\{\boldsymbol{0}\leq\boldsymbol{u}\leq\boldsymbol{N}+\boldsymbol{1}\}\setminus\{\boldsymbol{0}\leq\boldsymbol{u}\leq\boldsymbol{N}\}$
DO:}}{\footnotesize\par}
\item [{{\footnotesize\texttt{~}}}] {\footnotesize\texttt{~~~~~~SET
$\gamma:=\gamma+2^{-\Lambda(\boldsymbol{k})}|c_{\boldsymbol{k}}|^{2}$}}{\footnotesize\par}
\item [{{\footnotesize\texttt{~}}}] {\footnotesize\texttt{~~~~~~SET
$\beta:=\beta+2^{-\Lambda(\boldsymbol{k})}c_{\boldsymbol{k}}v_{\boldsymbol{k}}$}}{\footnotesize\par}
\item [{{\footnotesize\texttt{~}}}] {\footnotesize\texttt{~~~INCREMENT
$N_{h}$, $h=1,...,d$, by $1$}}{\footnotesize\par}
\item [{{\footnotesize\texttt{~}}}] {\footnotesize\texttt{RETURN $\beta$}}{\footnotesize\par}
\end{lyxlist}
\caption{\protect\label{alg:Application-of-Lemma}Approximate \foreignlanguage{english}{$\protect\int_{\mathbb{R}^{d}}v(\boldsymbol{x})f(\boldsymbol{x})d\boldsymbol{x}$
by the classical COS-i method, i.e., the sum $\sideset{}{'}\sum_{\boldsymbol{\boldsymbol{0}}\protect\leq\boldsymbol{k}\protect\leq\boldsymbol{N}}c_{\boldsymbol{k}}v_{\boldsymbol{k}}$.
To apply the damped COS-iv method, replace $v_{\boldsymbol{k}}$ by
$\tilde{v}_{\boldsymbol{k}}$ defined in Equation (\ref{eq:vkTilde}).
The number of terms $\boldsymbol{N}$ can be found during the evaluation
of the sum by} \foreignlanguage{english}{Inequality (\ref{eq:f_l-akek}).
The function $\Lambda$ is defined in Section \ref{sec:Notation}.}}
\end{algorithm}

\bibliographystyle{plainnat}
\bibliography{biblio}

\end{document}